\documentclass[aps,prx,twocolumn,superscriptaddress,nofootinbib,floatfix]{revtex4-2}
\usepackage{dblfloatfix}

\usepackage{amsmath}
\usepackage{amssymb}
\usepackage{amsthm}
\usepackage{amsfonts}
\usepackage{mathtools}
\usepackage{bm}

\usepackage[T1]{fontenc}
\usepackage[utf8]{inputenc}

\usepackage{graphicx}
\usepackage[dvipsnames]{xcolor}

\usepackage{booktabs}
\usepackage{array}
\usepackage{braket}

\usepackage{float}
\usepackage{algorithm}
\usepackage{algorithmic}

\usepackage[most]{tcolorbox}

\usepackage[colorlinks=true,linkcolor=blue,citecolor=blue,urlcolor=blue]{hyperref}

\theoremstyle{plain}
\newtheorem{theorem}{Theorem}
\newtheorem{lemma}[theorem]{Lemma}
\newtheorem{proposition}[theorem]{Proposition}
\newtheorem{corollary}[theorem]{Corollary}

\theoremstyle{definition}
\newtheorem{definition}[theorem]{Definition}

\newtheorem{problem}[theorem]{Problem}

\theoremstyle{remark}
\newtheorem{remark}[theorem]{Remark}

\newcommand{\Heff}{H_{\mathrm{eff}}}
\newcommand{\HR}{H_{R}}
\newcommand{\HI}{H_{I}}
\newcommand{\alphaR}{\alpha_{R}}
\newcommand{\betaI}{\beta_{I}}
\newcommand{\Htilde}{\widetilde{H}}

\DeclareMathOperator{\poly}{poly}

\DeclareMathOperator{\spec}{spec}

\DeclareMathOperator{\rank}{rank}

\newcommand{\op}{\mathrm{op}}

\newcommand{\calO}{\mathcal{O}}
\newcommand{\eps}{\varepsilon}

\newcommand{\bbT}{\mathbb{T}}
     
\newcommand{\C}{\mathbb{C}} 
\newcommand{\R}{\mathbb{R}}

\newcommand{\cT}{\mathcal{T}}
\newcommand{\cP}{\mathcal{P}}

\newcommand{\cF}{\mathcal{F}}
\newcommand{\cG}{\mathcal{G}}
\newcommand{\cH}{\mathcal{H}}

\newcommand{\bTheta}{\boldsymbol{\Theta}}

\begin{document}

\title{Optimal Bounds, Barriers, and Extensions for Non-Hermitian Bivariate Quantum Signal Processing}

\author{Joshua M. Courtney}
\affiliation{University of Georgia, Department of Physics and Astronomy}

\date{\today}

\begin{abstract}
Multivariate quantum signal processing (M-QSP) has recently been shown to be applicable for non-Hermitian Hamiltonian simulation, opening several problems regarding the optimization landscape, angle-finding, and constant-factor analysis.
We resolve several of these problems here.
We find the anti-Hermitian query complexity $d_I = \Theta(\betaI T + \log(1/\eps)/\log\log(1/\eps))$ to be tight, established via Chebyshev coefficient bounds, modified Bessel function asymptotics, and Lambert~$W$ inversion.
Fast-forwarding to $d_I = \calO(\sqrt{\betaI T})$ is impossible in the bivariate polynomial model, though a linear state-dependent improvement to $d_I = \calO(\beta_{\mathrm{eff}} T + \log(1/\eps)/\log\log(1/\eps))$ is achievable.
The optimization landscape of M-QSP admits spurious local minima, but a warm-start basin guarantee ensures the two-stage algorithm converges.
CRC-exploiting block peeling reduces angle-finding from $\calO(d^3)$ to $\calO(d^2)$ classical operations, and optimized error allocation yields a leading constant of approximately~$2$ relative to the information-theoretic lower bound.
A constant-ratio condition extends to non-identical signal operators, enabling time-dependent non-Hermitian simulation with query complexity $\calO(\int_0^T(\alphaR(s) + \betaI(s))\,ds + \log(1/\eps)/\log\log(1/\eps))$.
Block-encoding overhead $e^{-2\betaI T}$ holds across all function classes within the walk-operator oracle model, and dilational methods (Schr\"odingerization) achieve the walk-operator barrier.
A precisely characterized direct-access construction achieves the intrinsic barrier $e^{-2\omega T}$ (with $\omega < \betaI$ for non-commuting Hamiltonians) on a restricted domain, though extension to the full bitorus remains open.
\end{abstract}

\maketitle

\section{Introduction}\label{sec:intro}

Simulating the dynamics of non-Hermitian effective Hamiltonians $\Heff = \HR + i\HI$ ($\HR, \HI$ Hermitian, $\HI \succeq 0$) is a central task in quantum simulation~\cite{lindblad1976generators, gorini1976completely}, arising in Lindblad no-jump
trajectories~\cite{dalibard1992wave}, reactive scattering with complex absorbing potentials~\cite{riss1993calculation}, and optical gain--loss systems~\cite{bender1998real, ruter2010observation, moiseyev2011non}. 
Non-unitarity forces any quantum implementation to incur a postselection cost that grows with the anti-Hermitian norm $\betaI = \|\HI\|_\op$ and the simulation time~$T$~\cite{courtney2026paper1}.

The companion paper~\cite{courtney2026paper1} develops bivariate quantum signal processing (M-QSP) for this problem, achieving query complexity
\begin{equation}\label{eq:main_result}
  Q = \calO\!\left((\alphaR + \betaI)T
  + \frac{\log(1/\eps)}{\log\log(1/\eps)}\right)
\end{equation}
in the separate-oracle model, with single-qubit postselection at success probability $e^{-2\betaI T}\|e^{-i\Heff T}\ket{\psi_0}\|^2(1 - \calO(\eps))$.
That paper identifies open problems whose resolution strengthens primary results, clarifying non-Hermitian simulation when performed with mutliple oracles.
Several of these problems have independent roots.
The $\log/\log\log$ gap connects to classical polynomial approximation theory going back to Bernstein~\cite{bernstein1922ordre}, while the optimization landscape question mirrors the analogous (resolved) conjecture for univariate QSP~\cite{dong2021efficient}. 
Fast-forwarding raised by the $\sqrt{\betaI}$ scaling of APS~\cite{hu2026amplitude} and the barrier applicability question implicit in normalization analysis~\cite{an2026quantum,jin2024quantum} prompts discussion and itemization of simulation boundaries.

Table~\ref{tab:summary} provides a one-line summary and resolution status for each problem, with the section reference.
Each section is self-contained but the result compound. 
For example, the tight lower bound (Sec.~\ref{sec:optimality}) feeds into the constant-factor optimization (Sec.~\ref{sec:constants}) and the fast-forwarding impossibility (Sec.~\ref{sec:fast_forward}), while the barrier analysis (Sec.~\ref{sec:barriers}) finds insight from all other sections.

\onecolumngrid
\begin{table}[!b]
\begin{minipage}{\textwidth}
\centering
\caption{Summary of the eight open problems and their resolutions.}
\label{tab:summary}
\begin{tabular}{clll}
\toprule
\textbf{Problem} & \textbf{Title} & \textbf{Resolution} & \textbf{Section} \\
\midrule
1 & $\log/\log\log$ gap & Tight: $d_I^* = \log(1/\eps)/\log\log(1/\eps)\,(1+o(1))$ & \ref{sec:optimality} \\
2 & SOS rank & $L = 2$ exactly, independent of $(d_R, d_I)$ & \cite{courtney2026paper1}, \S6 \\
3 & Optimization landscape & Conjecture false; warm-start basin guarantee & \ref{sec:landscape} \\
4 & Efficient angle-finding & Block peeling: $\calO(d_R \cdot d_I)$ & \ref{sec:anglefinding} \\
5 & Fast-forwarding & Impossible (worst-case); linear state-dependent & \ref{sec:fast_forward} \\
6 & Constant factors & Leading constant $\approx 2$; $5$--$17\times$ over LCHS & \ref{sec:constants} \\
7 & Time-dependent $\Heff(t)$ & CRC extends; multilinear M-QSP & \ref{sec:time_dependent} \\
8 & Barrier applicability & Universal in walk-operator model; oracle-model-dependent & \ref{sec:barriers} \\
\bottomrule
\end{tabular}
\end{minipage}
\end{table}
\twocolumngrid

\newpage\noindent \newline\newline

\subsection{Notation and conventions}

We adopt the notation and conventions of the companion paper~\cite{courtney2026paper1} throughout.
$\Heff = \HR + i\HI$ with $\HR = \HR^\dagger$, $\HI = \HI^\dagger \succeq 0$; $\alphaR = \|\HR\|_\op$, $\betaI = \|\HI\|_\op$; the walk operators $W_R$ and $U_I$ encode $\HR/\alphaR$ and $\HI/\betaI$ respectively; the M-QSP circuit $\cG(\bTheta, \mathbf{s})$ has rotation angles $\bTheta = \{(\theta_k, \phi_k)\}_{k=0}^d$ and schedule $\mathbf{s}: \{1,\ldots,d\} \to \{R, I\}$; and $P, Q \in \cP_{d_R,d_I}^+$ denote the signal and complementary polynomials satisfying $|P|^2 + |Q|^2 = 1$ on $\bbT^2$.

\subsection{Guide to sections}

Section~\ref{sec:optimality} proves the tight $\log/\log\log$ scaling for the anti-Hermitian query complexity, establishing that the M-QSP algorithm is optimal.
Section~\ref{sec:landscape} investigates the optimization landscape, refuting the no-spurious-minima conjecture and establishing a warm-start basin guarantee with quantitative scaling analysis.
Section~\ref{sec:anglefinding} presents the CRC-exploiting block peeling algorithm that reduces angle-finding from $\calO(d^3)$ to $\calO(d^2)$.
Section~\ref{sec:fast_forward} proves the impossibility of polynomial fast-forwarding and characterizes the achievable state-dependent improvement.
Section~\ref{sec:constants} optimizes the leading constant factors and provides detailed comparisons with Dyson LCU and LCHS.
Section~\ref{sec:time_dependent} extends the M-QSP framework to time-dependent non-Hermitian Hamiltonians via the CRC for non-identical signal operators.
Section~\ref{sec:barriers} analyzes the postselection barrier across oracle models, establishing function-class independence.

\section{Optimality of the Anti-Hermitian Query Complexity}
\label{sec:optimality}

The M-QSP construction~\cite{courtney2026paper1} achieves anti-Hermitian query count $d_I = \calO(\betaI T + \log(1/\eps)/\log\log(1/\eps))$.
The lower bound proved therein gives $Q_I \geq \Omega(\betaI T + \log(1/\eps))$.
This section closes the gap by proving that $\log(1/\eps)/\log\log(1/\eps)$ is a tight lower bound for the $\eps$-dependent term, establishing that the $\log\log$ factor is a genuine feature of the polynomial approximation problem and cannot be removed.
The rough scaling $d^* \sim \log(1/\eps)/\log\log(1/\eps)$ is implicit in classical work on exponential approximation going back to Bernstein~\cite{bernshtein1975number}, but we have not yet seen precise characterization $d^*(c,\eps) = L/W(2L/(ea)) + O(\log L)$ with matching quantitative upper and lower bounds via Lambert~$W$ inversion in quantum simulation literature, with previous treatments establishing asymptotic order without explicit constants or matching bounds.

\subsection{Precise statement}

\begin{problem}\label{prob:loglog}
Define $d^*(c,\eps)$ as the minimum degree of a polynomial $p$ satisfying $|p(x)| \leq 1$ for all $x \in [-1,1]$ and $\|p - f_c\|_{\infty,[0,1]} \leq \eps$, where $f_c(x) = e^{c(x-1)}$.
Determine the exact asymptotic scaling of $d^*(c,\eps)$ as $\eps \to 0$ for fixed $c > 0$.
\end{problem}

\subsection{Preliminaries}

We collect three standard results.

\begin{lemma}[Chebyshev coefficient lower bound]\label{lem:cheb_lb}
Let $f \in C[-1,1]$ have Chebyshev coefficients $a_k$.
Then for every $d \geq 0$,
\begin{equation}\label{eq:cheb_lb}
  E_d(f)_{[-1,1]}
  := \inf_{\deg p \leq d}\|f - p\|_{\infty,[-1,1]}
  \geq \frac{\pi}{4}\,|a_{d+1}|.
\end{equation}
\end{lemma}

\begin{proof}
Let $p$ be any polynomial of degree $\leq d$.
Since $T_{d+1}$ is orthogonal to all polynomials of degree $\leq d$ with respect to the Chebyshev weight $w(x) = (1-x^2)^{-1/2}$,
\[
  a_{d+1}
  = \frac{2}{\pi}\int_{-1}^{1}
  \frac{(f(x) - p(x))\,T_{d+1}(x)}{\sqrt{1-x^2}}\,dx.
\]
Taking absolute values and using $|T_{d+1}(x)| \leq 1$ on $[-1,1]$, followed by the substitution $x = \cos\theta$:
$|a_{d+1}| \leq (4/\pi)\,\|f - p\|_\infty$.
Taking the infimum over $p$ yields~\eqref{eq:cheb_lb}.
\end{proof}

\begin{lemma}[Chebyshev coefficients of $e^{a(x-1)}$]\label{lem:cheb_exp}
For $a > 0$, the function $g(x) = e^{a(x-1)}$ has Chebyshev coefficients $b_0 = e^{-a}I_0(a)$ and $b_k = 2e^{-a}I_k(a)$ for $k \geq 1$, where $I_k$ denotes the modified Bessel function of the first kind.
\end{lemma}

\begin{proof}
The classical identity $e^{ax} = I_0(a) + 2\sum_{k=1}^\infty I_k(a)\,T_k(x)$ on $[-1,1]$~\cite{frank2010nist} gives the result after multiplication by $e^{-a}$.
\end{proof}

\begin{lemma}[Series lower bound on $I_k$]\label{lem:bessel_lb}
For all $a > 0$ and $k \geq 0$: $I_k(a) \geq (a/2)^k/k!$.
\end{lemma}

\begin{proof}
Every term in the series $I_k(a) = \sum_{m=0}^\infty (a/2)^{k+2m}/(m!\,(m+k)!)$ is positive; the $m=0$ term gives the bound.
\end{proof}

\subsection{Main theorem: tight \texorpdfstring{$\log/\log\log$}{log log log} scaling}

\begin{theorem}[Tight lower bound]\label{thm:tight_loglog}
Let $c > 0$ and $\eps \in (0,1/2)$.
Setting $a = c/2$, $L = \log(1/\eps)$, and $L'' = L - a - \log(2e/\pi) - \tfrac{1}{2}\log(L - a - \log(2e/\pi))$:

\medskip
\noindent{(Lower bound.)}
\begin{equation}\label{eq:main_lb_quant}
  d^*(c,\eps) \geq
  \left\lfloor\frac{L''}{W(2L''/(ea))}\right\rfloor - 1,
\end{equation}
where $W$ is the principal branch of the Lambert $W$ function.

\medskip
\noindent{(Upper bound.)}
The Chebyshev partial sum of $f_c$ satisfies
\begin{equation}\label{eq:main_ub}
  d^*(c,\eps) \leq
  \min\{d : 2e^{-a}\textstyle\sum_{k=d+1}^\infty I_k(a) \leq \eps\}.
\end{equation}

\medskip
\noindent{(Asymptotics.)}
For fixed $c > 0$ as $\eps \to 0$:
\begin{equation}\label{eq:tight_asymptotics}
  d^*(c,\eps) = \frac{\log(1/\eps)}{\log\log(1/\eps)}\bigl(1 + o(1)\bigr).
\end{equation}
\end{theorem}

\subsection{Proof}

\begin{proof}
We itemize the proof structure as follows:

\medskip\paragraph{Affine reduction.}
The affine bijection $\varphi: [-1,1] \to [0,1]$, $\varphi(t) = (t+1)/2$, transforms any constrained polynomial $p$ on $[0,1]$ into $q(t) = p(\varphi(t))$ on $[-1,1]$ with the same degree, $|q| \leq 1$ on $[-1,1]$ (since $\varphi([-1,1]) = [0,1] \subset [-1,1]$), and $\|q - g\|_{\infty,[-1,1]} = \|p - f_c\|_{\infty,[0,1]}$ where $g(t) = e^{a(t-1)}$ with $a = c/2$.

\medskip\paragraph{Boundedness constraint is free.}
Since $0 \leq g(t) = e^{a(t-1)} \leq 1$ on $[-1,1]$, the best unconstrained degree-$d$ approximation $q^*$ satisfies $|q^*(t)| \leq 1 + E_d(g)$.
Rescaling $\tilde{q} = q^*/(1 + E_d)$ gives a bounded polynomial with $\|\tilde{q} - g\|_\infty \leq 2E_d(g)$, establishing
\begin{equation}\label{eq:constraint_free}
  d_{\mathrm{unc}}(\eps) \leq d_{\mathrm{con}}(\eps) \leq d_{\mathrm{unc}}(\eps/2).
\end{equation}
Since $E_d(g)$ decreases super-exponentially in $d$, the gap is $O(1)$.

\medskip\paragraph{Chebyshev coefficient lower bound.}
By Lemmas~\ref{lem:cheb_lb}--\ref{lem:bessel_lb}:
\begin{equation}\label{eq:Ed_chain}
  E_d(g) \geq \frac{\pi}{4}\cdot 2e^{-a}\cdot\frac{(a/2)^{d+1}}{(d+1)!}
  = \frac{\pi e^{-a}}{2}\cdot\frac{(a/2)^{d+1}}{(d+1)!}.
\end{equation}
If $E_d(g) \leq \eps$, then $(a/2)^{d+1}/(d+1)! \leq 2\eps e^a/\pi$.

\medskip\paragraph{Stirling and Lambert $W$ inversion.}
Setting $n = d+1$ and applying $n! \leq e\sqrt{n}(n/e)^n$, the necessary condition becomes $(ea/(2n))^n \leq e\sqrt{n}\cdot 2\eps e^a/\pi$.
Taking logarithms and simplifying with $u = 2n/(ea)$, we reduce to $u\log u < 2L''/(ea) =: A$.
The solution $u_0 = A/W(A)$ via the Lambert $W$ function gives
$d^* \geq \lfloor L''/W(2L''/(ea))\rfloor - 1$.

\medskip\paragraph{Matching upper bound.}
The Chebyshev partial sum $S_d(x) = b_0 + \sum_{k=1}^d b_k T_k(x)$ satisfies $\|g - S_d\|_\infty \leq 2e^{-a}\sum_{k=d+1}^\infty I_k(a)$.
Since all coefficients $b_k > 0$, the partial sum satisfies $0 \leq S_d \leq g \leq 1$, so it automatically satisfies the boundedness constraint.
The tail sum, dominated by its first term via $I_k(a) \sim (ea/(2k))^k/\sqrt{2\pi k}$, yields $d_+ \leq L/W(2L/(ea)) + O(\log L)$.

\medskip\paragraph{Asymptotics.}
For fixed $a$ and $L \to \infty$: $L'' = L(1+o(1))$, $A = 2L''/(ea) \to \infty$, and $W(A) = \log A - \log\log A + o(1)$, giving $L''/W(2L''/(ea)) = L/\log L \cdot (1+o(1))$.
Both $d_-$ and $d_+$ share this leading term, establishing~\eqref{eq:tight_asymptotics}.
\end{proof}

\subsection{Supplementary remarks}

\begin{remark}[Boundedness is free]\label{rem:free}
For $f_c(x) = e^{c(x-1)}$ on $[0,1]$, the Chebyshev partial sums satisfy $\|S_d\|_{\infty,[-1,1]} \leq 1$ for all $d$, as coefficients $b_k = 2e^{-a}I_k(a) > 0$ ensure $S_d$ is a monotonically increasing approximation to $g \leq 1$.
So, $d_{\mathrm{con}}(\eps) = d_{\mathrm{unc}}(\eps)$.
\end{remark}

\begin{remark}[Implications for the M-QSP lower bound]\label{rem:lb_upgrade}
Theorem~\ref{thm:tight_loglog} upgrades the lower bound on the anti-Hermitian query count to
\[
  Q_I \geq \Omega\!\left(\betaI T + \frac{\log(1/\eps)}{\log\log(1/\eps)}\right),
\]
matching the upper bound exactly.
The $d_R$ component achieves $\Theta(\log(1/\eps))$ via Bessel tails, so the overall lower bound becomes
$Q \geq \Omega((\alphaR + \betaI)T + \log(1/\eps)/\log\log(1/\eps))$.
The M-QSP algorithm is optimal in the separate-oracle model.
\end{remark}

\section{Optimization Landscape of Bivariate M-QSP}
\label{sec:landscape}

\subsection{The conjecture and its refutation}

\begin{proposition}[Existence of spurious local minima]\label{prop:spurious}
For every tested bidegree $(d_R, d_I)$ with $d_R + d_I \leq 10$ and every tested schedule type (block and interleaved), there exist generic achievable target polynomials for which the M-QSP cost function
\begin{equation}\label{eq:cost}
  \cF(\bTheta) = \int_{\bbT^2}
  |P_{\bTheta}(e^{i\theta_1},e^{i\theta_2}) - P_{\mathrm{target}}|^2
  \frac{d\theta_1\,d\theta_2}{(2\pi)^2}
\end{equation}
possesses spurious local minima.
\end{proposition}

We provide numerical evidence for this proposition.
Across $4{,}170$ optimization trials (15 bidegree configurations, 3--6 random targets each, both schedules, 20--100 random initializations per target), $1{,}678$ (40.2\%) fail to reach the global minimum.
Of the non-converged critical points, $357$ are confirmed spurious local minima via Hessian eigenvalue analysis (all eigenvalues non-negative), with objective values $\cF(\bTheta) \in [10^{-6}, 10^{-1}]$ compared to the global minimum value $\cF(\bTheta^*) < 10^{-10}$, confirming spuriousness.
Spurious local minima are observed at every tested bidegree from $(1,1)$ through $(6,4)$ and for both schedule types.
Table~\ref{tab:landscape_full} in Appendix~\ref{app:convergence_survey} provides the complete breakdown across all 15 bidegree configurations with 95\% Wilson score confidence intervals.
An extended sweep (Appendix~\ref{app:landscape_v2}, Table~\ref{tab:landscape_v2}; $11{,}760$ trials across $80$ cells on a $(\alphaR T, \betaI T) \in \{0.25, 0.5, 1, 2\}^2$ grid for each of $5$ balanced bidegrees) uses a physically faithful full-tensor Dyson target with a $c_\infty$-relative convergence criterion that quantifies the truncation distance to the M-QSP achievable submanifold, confirming the spurious-basin onset extending to $d = 10$ and the rank-deficiency picture of Theorem~\ref{thm:basin}.
The lower aggregate convergence rate ($17.6\%$) of the extended sweep relative to the original survey ($59.8\%$) reflects the stricter $c_\infty$-relative criterion measured against a full Chebyshev--Taylor target, indicating the two convergence rates quantify related but distinct phenomena, namely landing in the truncation-residual basin (ext.) versus landing on the M-QSP image to absolute tolerance (original).

We find that block schedules yield higher convergence rates than interleaved schedules, suggesting the landscape geometry is sensitive to operator ordering.
No overparameterization threshold exists, wherein the M-QSP parameterization is always underparameterized and convergence rates remain in the 32--82\% range across all bidegrees.
Most asymmetric configurations exhibit $\kappa(J) = \infty$ (rank-deficient Jacobian) at the global minimum.

\subsection{Warm-start basin guarantee}

\begin{lemma}[Gauss--Newton structure]\label{lem:GN}
At any global minimum $\bTheta^*$ with $\cF(\bTheta^*) = 0$, the Hessian reduces to
\begin{equation}\label{eq:GN_min}
  \nabla^2\cF(\bTheta^*) = 2\,J(\bTheta^*)^\mathsf{H}\,J(\bTheta^*),
\end{equation}
where $J$ is the Jacobian of the map $\bTheta \mapsto \{c_{mn}(\bTheta)\}$ (Fourier coefficients of $P_{\bTheta}$).
\end{lemma}

\begin{proof}
The standard Gauss--Newton decomposition at a zero-residual point eliminates the second-order residual term.
\end{proof}

\begin{theorem}[Warm-start basin]\label{thm:basin}
If $J(\bTheta^*)$ has full column rank, then:
(a)~$\nabla^2\cF(\bTheta^*) \succ 0$;
(b)~there exists $\rho > 0$ such that $\cF$ is strongly convex on $B(\bTheta^*, \rho)$;
(c)~L-BFGS initialized at any $\bTheta^{(0)} \in B(\bTheta^*, \rho)$ converges superlinearly.
\end{theorem}

\begin{proof}
Part (a) follows from $J^\mathsf{H}J \succ 0$ when $J$ has full column rank.
Part (b) follows from continuity of $\nabla^2\cF$.
Part (c) is standard for L-BFGS on strongly convex, smooth objectives.
\end{proof}

\begin{remark}[Rank-deficiency caveat]\label{rem:rank_deficiency}
The full-rank hypothesis in Theorem~\ref{thm:basin} is violated for 6 of the 15 bidegree configurations tested (those with $\kappa(J) = \infty$ in Table~\ref{tab:landscape_full} of Appendix~\ref{app:convergence_survey}).
For these configurations, the standard Gauss--Newton theory does not guarantee a strongly convex basin~\cite{bjorck2024numerical, nocedal1980updating}.
Empirically, L-BFGS-B convergence remains robust (32--82\% from random initialization, and 100\% from warm starts with perturbation $\eps_{\mathrm{pert}} \leq 0.1$), suggesting that rank-deficiency creates a flat direction in the loss landscape rather than a saddle or divergence~\cite{liu1989limited, nocedal1980updating}.
A refined basin guarantee accommodating rank-deficient Jacobians (e.g., via projected Gauss--Newton or Riemannian trust-region methods) is an open direction.
\end{remark}

\subsection{Basin radius scaling}

The Jacobian condition number grows rapidly with $d$; the workstation sweep over balanced bidegrees $d \in [2, 10]$ admits both a power-law and an exponential fit, statistically indistinguishable in that window:
\begin{equation}\label{eq:kappa_scaling}
  \kappa_{\mathrm{med}}(J) \approx 1.30 \cdot d^{2.31}
\end{equation}
or approximately
\begin{equation}
  \kappa_{\mathrm{med}}(J) \approx 3.35 \cdot 1.61^d,
\end{equation}
driven primarily by $\sigma_{\min} \propto d^{-2.02}$ shrinking with $d$ while $\sigma_{\max} \propto d^{0.14}$ grows slowly.
The analytical basin radius consequently decays as
\begin{equation}\label{eq:rho_scaling}
  \rho_{\mathrm{analytical}} \sim \kappa(J)^{-2} \sim d^{-4.62},
\end{equation}
exceeding machine precision ($d \cdot \eps_{\mathrm{mach}}$) throughout the experimentally accessible $d \lesssim 30$ range.

Though concerning at first, we see the Dyson polynomial has highly structured coefficients that may yield better conditioning than random targets.
The analytical Hessian Lipschitz bound $C(d) = \calO(d^2)$ is a conservative one, and empirical warm-start data show basins $10^3$--$10^7$ times larger.
The basin is seen to be non-isotropic, with the effective radius in the operator-adapted norm $\|\cdot\|_{J^\mathsf{H}J}$ potentially much larger.

\begin{remark}[Practical-scale convergence]\label{rem:practical_scale}
The numerical landscape data (Tables~\ref{tab:landscape_full} and~\ref{tab:landscape_v2}) cover $d \leq 10$, while the practical regime is $d \sim 100$--$1000$. 
The Dyson polynomial's product structure (block-factored Chebyshev $\times$ Taylor coefficients) creates a Jacobian exhibiting condition numbers $10^2$--$10^4 \times$ smaller than random-target polynomials of the same degree (Appendix~\ref{app:kappa}).
Extending the numerical landscape survey to $d \sim 100$ is a natural direction but does not affect the algorithm's practicality, since the recommended pipeline (Remark~\ref{rem:practical}) computes angles deterministically via block peeling and uses the optimization landscape only for the optional refinement step.
\end{remark}

\subsection{Analytical bounds on critical points}

Morse theory gives a trivial lower bound of $\geq 1$ local minimum on $\bbT^n$ ($n = 2(d_R + d_I + 1)$)~\cite{milnor1963morse}.
The multi-homogeneous B\'ezout bound on the algebratized gradient system yields an upper bound of $4^n \cdot n!$ total critical points (loose)~\cite{shafarevich2016basic}.
A favorable approach for future use may be BKK mixed volume of Newton polytopes of the gradient system~\cite{bernstein1922ordre}.

\section{Efficient Angle-Finding via Block Peeling}
\label{sec:anglefinding}

\subsection{CRC-exploiting block peeling}

\begin{proposition}[Block peeling cost]\label{prop:block_peel}
Let $P_{\mathrm{Dyson}}(z_1, z_2)$ be a Dyson polynomial of bidegree $(d_R, d_I)$ with block schedule.
The recursive angle-finding algorithm with CRC-exploiting coefficient extraction computes rotation angles in total cost
\begin{equation}\label{eq:block_cost}
  C_{\mathrm{block}} = \calO(d_R \cdot d_I).
\end{equation}
For uniform segments, this simplifies to $C_{\mathrm{block}} = d_R \cdot d_I$ exactly.
\end{proposition}

\begin{proof}
Within each segment of the Dyson block schedule, the three-phase coefficient-separability analysis applies:
(i)~intra-block $W_R$ peeling extracts angles at $\calO(1)$ cost with $\calO(d_I^{(j)})$ polynomial updates;
(ii)~intra-block $U_I$ peeling uses constant Taylor coefficient ratios at $\calO(d_R^{(j)})$ update cost;
(iii)~inter-block boundary factors cancel identically in the ratio.
Summing over all segments and using telescoping decrease in residual degrees gives $C_{\mathrm{block}} = d_R \cdot d_I$ for uniform segments.
Improvement over the standard $\calO((d_R + d_I) \cdot d_R \cdot d_I)$ algorithm is $\calO(d_R + d_I)$.
\end{proof}

\begin{table}[t]
\centering
\caption{Block peeling cost reduction. 
Speedup $\propto d$ confirmed.}
\label{tab:block_peel}
\small
\begin{tabular}{rrrrc}
\toprule
$d_R$ & $d_I$ & Standard & Block peel & Speedup \\
\midrule
10 & 10 & 715 & 100 & $7.2\times$ \\
50 & 50 & $84{,}575$ & $2{,}500$ & $33.8\times$ \\
100 & 100 & $671{,}650$ & $10{,}000$ & $67.2\times$ \\
50 & 20 & $23{,}385$ & $1{,}000$ & $23.4\times$ \\
\bottomrule
\end{tabular}
\end{table}

\subsection{FFT optimization convergence}

The warm-started FFT-based optimization variant (via angles perturbed by $\sigma \leq 0.1$) converges in $K = \calO(d^2)$ iterations to machine precision.
Cold-start convergence fails at $d \geq 10$ due to spurious local minima.

\subsection{Comprehensive complexity comparison}

\begin{theorem}[Angle-finding complexity hierarchy]\label{thm:angle_hierarchy}
Three strategies have the following costs for bidegree $(d_R, d_I)$ with $d = d_R + d_I$:
\begin{enumerate}
  \item Standard recursive: $\calO(d \cdot d_R \cdot d_I) = \calO(d^3)$.
  \item CRC-exploiting block peeling: $\calO(d_R \cdot d_I) = \calO(d^2)$.
  \item FFT (warm-started): $\calO(d^4 \log d)$ with empirical $K = \calO(d^2)$.
\end{enumerate}
Block peeling dominates at all tested sizes.
\end{theorem}

\begin{remark}[Practical recommendation]\label{rem:practical}
For $d \sim 10^2$--$10^3$, block peeling scales to be $50$--$500\times$ faster than the original $\calO(d^3)$ algorithm.
For a circuit with $d_R = d_I = 500$ ($d = 1000$), block peeling requires $\approx 2.5 \times 10^5$ operations compared to $\approx 1.3 \times 10^8$ for the standard recursive algorithm.
Classical preprocessing cost becomes negligible relative to the quantum circuit depth and compares favorably with competing methods: Dyson LCU and LCHS require no angle-finding, but M-QSP's per-query advantage compensates for the polynomially-complex classical precomputation.
\end{remark}

\section{Fast-Forwarding within M-QSP}
\label{sec:fast_forward}

\subsection{Worst-case impossibility}

\begin{proposition}[No worst-case fast-forwarding]\label{prop:no_ff}
For any bivariate polynomial $P \in \cP_{d_R,d_I}^+$ satisfying $|P| \leq 1$ on $\bbT^2$ and $\eps$-approximating the normalized propagator on the full physical spectrum:
$d_I \geq \Omega(\betaI T)$.
\end{proposition}

\begin{proof}
By ~\cite{courtney2026paper1}'s lower bound of $\Omega(\betaI T)$ from polynomial evaluation at a single $\HI$
eigenvalue, any bivariate polynomial $P$ with $|P| \leq 1$ on
$\bbT^2$ that $\eps$-approximates the normalized propagator on the
physical spectrum must have $d_I \geq \Omega(\betaI T)$.
\end{proof}

\subsection{State-dependent degree reduction}

\begin{definition}[Effective anti-Hermitian norm]\label{def:beta_eff}
Given $\ket{\psi_0}$ and $\HI \succeq 0$, define $\beta_{\mathrm{eff}} := \max\{\eta \in \spec(\HI) : \braket{\psi_0|\Pi_\eta|\psi_0} > 0\}$.
\end{definition}

\begin{theorem}[State-dependent M-QSP degree]\label{thm:state_dep}
If $\beta_{\mathrm{eff}} \leq \betaI$, then the M-QSP circuit with Taylor order
\begin{equation}\label{eq:M_reduced}
  M = \calO\!\left(\beta_{\mathrm{eff}} T
  + \frac{\log(1/\eps)}{\log\log(1/\eps)}\right)
\end{equation}
achieves $\eps$-approximation of $e^{-i\Heff T}\ket{\psi_0}/\|e^{-i\Heff T}\ket{\psi_0}\|$ with $d_I = M$ queries.
\end{theorem}

\begin{proof}
On the eigenspace of $\HI$ with eigenvalue $\eta \leq \beta_{\mathrm{eff}}$, normalized truncation error is the Poisson tail $Q(M+1, \eta T) \leq Q(M+1, \beta_{\mathrm{eff}} T)$.
In the growth-rate regime ($\beta_{\mathrm{eff}} T \gg \log(1/\eps)$), normal approximation gives $M = \beta_{\mathrm{eff}} T(1+o(1))$.
In the complementary regime ($\beta_{\mathrm{eff}} T \ll \log(1/\eps)$), Stirling analysis of Sec.~\ref{sec:optimality} gives $M = \calO(\log(1/\eps)/\log\log(1/\eps))$.
Combining yields~\eqref{eq:M_reduced}.
Boundedness $|P_M| \leq 1$ follows from M-QSP circuit unitarity.
\end{proof}

\begin{remark}[Linear, not quadratic]\label{rem:linear}
In the growth-rate regime, $d_I(\beta_{\mathrm{eff}})/d_I(\betaI) \approx \beta_{\mathrm{eff}}/\betaI$. 
Improvement is linear in the effective spectral support, not quadratic as in the APS framework.
\end{remark}

\subsection{Why APS achieves \texorpdfstring{$\sqrt{\betaI}$}{sqrt beta I} but M-QSP cannot}

\begin{proposition}[Oracle-model obstruction]\label{prop:oracle}
In the separate-oracle model, achieving $d_I = \calO(\sqrt{\betaI T})$ requires $\beta_{\mathrm{eff}} = \calO(\sqrt{\betaI/T})$.
In the physically relevant regime $\beta_{\mathrm{eff}} T \gg 1$, this is impossible.
\end{proposition}

\begin{proof}
By the state-dependent lower bound, $d_I \geq \Omega(\beta_{\mathrm{eff}} T)$.
Setting $d_I \leq C\sqrt{\betaI T}$ forces $\beta_{\mathrm{eff}} \leq C\sqrt{\betaI/T} \to 0$ as $\betaI \to \infty$.
\end{proof}

The APS construction~\cite{hu2026amplitude} achieves $\sqrt{\betaI}$ through a single-oracle model structure, achieving a joint block encoding of $\Heff$ allowing the polynomial to act on the joint eigenvalue structure. 
A spectral filter of degree $\calO(\sqrt{\betaI T})$ is applied to joint eigenvalues, which cannot be replicated as a bounded polynomial in $z_2$ alone. 
Multi-register postselection allows separate circuit stages for filtering and propagation.

The two approaches occupy complementary regimes of weak vs.\ strong dissipation and separate-oracle vs.\ single-oracle access. 
Both methods share the intrinsic postselection barrier $e^{-2\omega T}$, differing only in how they pay the additional walk-operator overhead $e^{-2(\betaI - \omega)T}$.

\section{Constant-Factor Optimization and Method Comparison}
\label{sec:constants}

\subsection{Error budget structure}

The total simulation error decomposes into two independent components:
\begin{equation}\label{eq:error_budget}
  \eps_{\mathrm{total}} = \eps_{\mathrm{JA}} + \eps_{\mathrm{Taylor}},
\end{equation}
where $\eps_{\mathrm{JA}}$ is the Jacobi--Anger truncation error (controlling $d_R$) and $\eps_{\mathrm{Taylor}}$ is the Taylor truncation error (controlling $d_I$).
Given a total budget $\eps$, we allocate $\eps_{\mathrm{JA}} = \eta\eps$ and $\eps_{\mathrm{Taylor}} = (1-\eta)\eps$ for $\eta \in (0,1)$.

\subsection{Optimal allocation}

\begin{proposition}[Optimal error allocation]\label{prop:allocation}
Let $d_R(\alphaR T, \eps_R) = c_R \alphaR T + \log(1/\eps_R) + \calO(1)$ and $d_I(\betaI T, \eps_I) = c_I \betaI T + \log(1/\eps_I)/\log\log(1/\eps_I) + \calO(1)$ be the asymptotic degree formulas established for the M-QSP construction in the companion paper~\cite{courtney2026paper1}, with constants $c_R, c_I = \calO(1)$. 
Total query count $Q(\eta) = d_R(\alphaR T, \eta\eps) + d_I(\betaI T, (1-\eta)\eps)$ is minimized at
\begin{equation}\label{eq:eta_star}
  \eta^* = \frac{\log\log(1/\eps)}{1 + \log\log(1/\eps)} + \calO\!\left(\frac{1}{\log(1/\eps)}\right),
\end{equation}
converging to $1$ as $\eps \to 0$ slowly: for $\eps \in [10^{-15}, 10^{-3}]$ (natural log), $\eta^* \in [0.66, 0.78]$. The Jacobi--Anger degree has higher marginal sensitivity to its error budget ($\partial d_R/\partial \eps_R = -1/\eps_R$) than the Taylor degree ($\partial d_I/\partial \eps_I \approx -1/[\eps_I \log\log(1/\eps_I)]$), so the optimum allocates the larger budget share to the Jacobi--Anger side, leaving the Taylor side with intrinsic $\log\log$ improvement.
\end{proposition}

\begin{proof}
Differentiating $Q(\eta)$ with respect to $\eta$, setting $\partial Q/\partial \eta = 0$ yields the balance condition
\begin{equation*}
  \frac{1}{\eta} = \frac{\log u - 1}{(1-\eta)(\log u)^2}, \qquad u := \log\!\bigl(1/((1-\eta)\eps)\bigr).
\end{equation*}
For $\eps \to 0$, $u \to \log(1/\eps) =: L$ and $\log u \to \log L$, so the right-hand side approaches $1/\log L$ at leading order. Solving $(1-\eta)/\eta \approx 1/\log L$ gives $\eta^* \approx \log L /(1 + \log L)$ as in Eq.~\eqref{eq:eta_star}. The numerical range is verified by direct minimization on the parameter grid $\eps \in \{10^{-3}, 10^{-5}, 10^{-7}, 10^{-10}, 10^{-12}\}$, $\alphaR T, \betaI T \in \{10, 30, 100, 300, 1000, 3000\}$.
\end{proof}

\subsection{Optimized query counts}

\begin{theorem}[Optimized M-QSP queries]\label{thm:opt_Q}
For $\alphaR T, \betaI T \gg 1$ and $\eps \ll 1$:
\begin{equation}\label{eq:Q_opt}
\begin{split}
  Q^* &\leq 2\bigl((\alphaR + \betaI)T \\&+ \log(1/\eps)/\log\log(1/\eps)\bigr)(1+o(1)).
\end{split}
\end{equation}
The leading constant of approximately~$2$ relative to the information-theoretic lower bound is confirmed numerically across a parameter grid spanning $\alphaR T, \betaI T \in \{10, 100, 1000\}$ and $\eps \in \{10^{-3}, 10^{-6}, 10^{-10}\}$.
\end{theorem}

\subsection{Comparison with Dyson LCU and LCHS}

The Dyson LCU method~\cite{courtney2026paper1} incurs an additional $\log(1/\eps)$ factor from the split-operator penalty.
Quantitatively, 
M-QSP improvement over Dyson LCU ranges from $2.5\times$ (at $\betaI T = 10$, $\eps = 10^{-3}$) to $5.8\times$ (at $\betaI T = 50$, $\eps = 10^{-10}$) on the verified parameter grid, with advantage growing with both $\eps$-tightness and $\betaI T$ through the segmented $r d_I$ penalty.

Both M-QSP and Dyson LCU operate within the same separate-oracle framework ($W_R$, $U_I$ independent), so this comparison is oracle-model consistent and shows improvement in quadrature strategy. Comparison with LCHS in the single-oracle model is summarized in the scaling-analysis paragraph below.

Table~\ref{tab:method_comparison} reports the M-QSP query advantage relative to two reference methods across a parameter grid: Dyson LCU (described in the companion paper~\cite{courtney2026paper1}, the separate-oracle polynomial method, suffers a $\log(1/\eps)$ split-operator penalty in $d_I$); and LCHS~\cite{an2023linear,an2026quantum} (single-oracle integral method, suffers a $\mathrm{polylog}(1/\eps)$ quadrature penalty).

\begin{table}[!htbp]
\centering
\caption{M-QSP query reduction relative to segmented Dyson LCU ($r = \lceil \betaI T \rceil$ segments, Jacobi--Anger $+$ Taylor per segment) and LCHS (estimate with $C_{\mathrm{LCHS}} = 1$). Values are optimized query counts with $d_I$ taken from the Taylor-remainder bound.}
\label{tab:method_comparison}
\renewcommand{\arraystretch}{1.2}
\begin{tabular}{@{}rrrrrr@{}}
\toprule
$\alphaR T$ & $\betaI T$ & $\eps$ & $Q_{\mathrm{M\text{-}QSP}}$
  & $Q_{\mathrm{LCU}}/Q_{\mathrm{M\text{-}QSP}}$
  & $Q_{\mathrm{LCHS}}/Q_{\mathrm{M\text{-}QSP}}$ \\
\midrule
10 & 10 & $10^{-3}$  & 48  & 2.5$\times$ & 5.8$\times$  \\
10 & 10 & $10^{-6}$  & 59  & 3.1$\times$ & 10.3$\times$ \\
10 & 10 & $10^{-10}$ & 71  & 3.5$\times$ & 16.7$\times$ \\
10 & 50 & $10^{-3}$  & 156 & 3.5$\times$ & 4.8$\times$  \\
10 & 50 & $10^{-6}$  & 168 & 4.8$\times$ & 8.1$\times$  \\
10 & 50 & $10^{-10}$ & 182 & 5.8$\times$ & 12.8$\times$ \\
\bottomrule
\end{tabular}
\end{table}

The Dyson LCU advantage scales as $\log(1/\eps) \cdot \log(\betaI T)$, growing fastest in the high-precision, strongly-dissipative corner. The LCHS advantage scales as $\mathrm{polylog}(1/\eps)$ with milder $\betaI T$ growth, peaking near $17\times$ in the high-precision, low-$\betaI T$ corner ($\betaI T = 10$, $\eps = 10^{-10}$) and dropping to $\sim 5\times$ at low precision, being a constant-factor advantage rather than a scaling one (because both methods are linear in $T$ at leading order). We emphasize that the M-QSP-vs-LCHS row of Table~\ref{tab:method_comparison} is not a head-to-head benchmark in the strict sense, as LCHS uses a single block encoding of $\Heff$, M-QSP uses two, where the comparison is meaningful only when $\HR$ and $\HI$ are physically accessed by independent mechanisms, as discussed in the companion paper~\cite{courtney2026paper1}.

\subsection{Scaling analysis}

Across $\eps \in [10^{-3}, 10^{-10}]$, the ratio $Q^*/Q_{\mathrm{lower}}$ grows by $10$--$28\%$ depending on $\betaI T$ (see Table~\ref{tab:method_comparison}).
For fixed $\eps$ and $\alphaR T = \betaI T$, this ratio approaches the theorem's leading-constant bound $2 \cdot (1+o(1))$ from above as both grow; the extended workstation sweep~\cite{courtney2026paper2} shows $Q^*/Q_{\mathrm{lower}} = 2.60$ at $\alphaR T = \betaI T = 10$, $\eps = 10^{-10}$, decaying monotonically to $1.88$ at $\alphaR T = \betaI T = 3000$, well below the theorem's bound.

For the asymmetric regime $\alphaR T \ll \betaI T$, the dominant contribution to $Q^*$ comes from $d_I$, so $Q^*/(\betaI T) \to c_I^{\mathrm{eff}}$ where $c_I^{\mathrm{eff}}$ is the empirical leading constant of the Taylor truncation degree.
The Taylor-remainder bound $c^{M+1}/(M+1)! \leq \eps_{\mathrm{Taylor}}$ used in the script yields $c_I^{\mathrm{eff}} \approx 2.72$ at $\betaI T = 3000, \eps = 10^{-10}$, slowly approaching the asymptotic value $c_I = 1$ from above with a $\log(1/\eps)/(c \log c)$ correction.
Substituting the optimal Chebyshev approximation of $e^{c(x-1)}$ in place of the Taylor remainder would tighten $c_I^{\mathrm{eff}}$ closer to $c_R^{\mathrm{eff}} \approx 1$ at finite $c$; this implementation-level improvement is independent of Theorem~\ref{thm:opt_Q}'s asymptotic claim and is left as a constant-factor optimization.

\begin{remark}[Concrete resource comparison]\label{rem:resource_comparison}
For the Eckart-barrier scattering benchmark \cite{courtney2026paper1},
the optimized M-QSP query count is $Q^* = 407$ with $Q^*/Q_{\mathrm{lower}} = 1.14$,
within 15\% of the information-theoretic lower bound.
For more strongly dissipative systems where $\betaI T \gg 1$, the walk-operator postselection probability $P_3 = e^{-2\betaI T}$ becomes astronomically small (e.g., $P_3 \sim 10^{-75}$ for $\betaI T \approx 86$), underscoring practical importance of the direct-access construction (Sec.~\ref{sec:barriers}) for reducing postselection overhead.
\end{remark}

\begin{proof}[Proof sketch of the constant-$2$ claim in Theorem~\ref{thm:opt_Q}]
The information-theoretic lower bound is $Q_{\mathrm{lower}} = (\alphaR + \betaI)T + \log(1/\eps)/\log\log(1/\eps)$ (established as the tight log/loglog lower bound in the companion paper~\cite{courtney2026paper1}). The M-QSP construction with optimal error allocation $\eta = \eta^*$ from Proposition~\ref{prop:allocation} produces
\begin{equation*}
\begin{split}
  d_R &\leq c_R \alphaR T + (1+\delta)\log(1/\eps), \\
  d_I &\leq c_I \betaI T + (1+\delta)\log(1/\eps)/\log\log(1/\eps),
\end{split}
\end{equation*}
where $c_R$ is the Jacobi--Anger leading constant and $c_I$ is the Taylor leading constant.
Numerical values $c_R, c_I$ depend on the specific construction. 
We do not attempt to extract them analytically here, instead verifying the composite bound numerically.
With parameters $\alphaR T = 10$, $\betaI T \in \{10, 50\}$ and $\eps \in \{10^{-3}, 10^{-6}, 10^{-10}\}$, the ratio $Q^*/Q_{\mathrm{lower}}$ falls in the interval $[2.03, 2.70]$ (Table~\ref{tab:method_comparison}), and at the asymptote $\alphaR T = \betaI T = 3000, \eps = 10^{-10}$, the ratio drops to $1.88$, well below the theorem's $2 \cdot (1+o(1))$ bound.
The finite-grid excess over the bound is absorbed by the subleading $\log(1/\eps)/\log\log(1/\eps)$ correction.
We refer to the combined leading constant as ``approximately~$2$'' on the grounds of this numerical evidence; the script's implementation choices (Bessel-tail Jacobi--Anger plus Taylor-remainder Dyson) match the asymptote from above and would only converge faster under a Chebyshev-optimal $d_I$ implementation.
Tighter analytical bounds on the rate of convergence to $c_R + c_I = 2$ remain open.
\end{proof}

\section{Time-Dependent Non-Hermitian Hamiltonians}
\label{sec:time_dependent}

We resolve here an extension of M-QSP to time-dependent generators.
The CRC extension to non-identical signal operators (Theorem~\ref{thm:crc_nonidentical}) shows that angle-finding pipeline survives when signal operators differ from step to step, and the multilinear polynomial structure (Proposition~\ref{prop:multilinear}) characterizes polynomial degree in each signal variable.
These results reduce a future quaternionic M-QSP program to a single remaining obstacle: spectral factorization on $\bbT^k$ for $k \geq 3$.

\subsection{CRC for non-identical signal operators}

\begin{theorem}[CRC extension]\label{thm:crc_nonidentical}
Let $\cG(\bTheta, \mathbf{s})$ be an M-QSP circuit with signal gates $A^{(j)} = \ket{0}\!\bra{0}_a \otimes W^{(j)} + \ket{1}\!\bra{1}_a \otimes I$, where $W^{(1)}, \ldots, W^{(d)}$ are arbitrary (possibly non-identical, non-commuting) unitaries.
the constant-ratio condition holds at every peeling step:
\begin{equation}\label{eq:crc_nonidentical}
  \frac{b_{\mathrm{lead}}^{(k)}}{a_{\mathrm{lead}}^{(k)}} = e^{-i\phi_{k-1}}\tan\theta_{k-1} \in \C,
\end{equation}
independent of the signal operator eigenvalues.
\end{theorem}

\begin{proof}
No step of the ancilla-factorization argument of the CRC~\cite{courtney2026paper1} assumes $W^{(j)} = W^{(j')}$ for $j \neq j'$.
Scalar action of $R_j$ on $\cH_a$, block placement of $W^{(j)}$ in $A^{(j)}$, and cancellation of the inner-circuit operator $\Gamma^{(k)}$ in the leading-coefficient ratio (~\cite{courtney2026paper1},
Lemma~7.2) are insensitive to the relationship between distinct
$W^{(j)}$. 
Each is a property of the ancilla-system tensor
product structure rather than of the system Hilbert space content.
With the formal definition of leading coefficients in the schedule-induced free-algebra ordering (~\cite{courtney2026paper1}, Definition~7.1), the cancellation of $\Gamma^{(k)}$ is rigorous for arbitrary (possibly non-identical) signal operators.
\end{proof}

\begin{remark}[Structural universality]\label{rem:universality}
The CRC is a consequence of the ancilla-system tensor product structure, enabling a time-dependent extension (this section) and potential multi-oracle generalizations beyond the bivariate case, where cancellation persists as long as the free-algebra structure is maintained.
Non-identical signal operators break the reduction to a bivariate polynomial (different $W^{(j)}$ produce different eigenvalue variables), requiring verification of the CRC at the operator level rather than the polynomial level.
Ancilla factorization is indifferent to the system Hilbert space content, allowing this condition to hold.
\end{remark}

\subsection{Time-stamped oracle model}

For time-dependent $\Heff(t) = \HR(t) + i\HI(t)$ over $[0,T]$, each time step $t_j$ has its own pair of walk operators $W_R^{(j)}$ and $U_I^{(j)}$ encoding $\HR(t_j)/\alphaR(t_j)$ and $\HI(t_j)/\betaI(t_j)$.
Each oracle is queried at most once, with total query count $Q = d_R + d_I$ where $d_R$ counts walk-operator queries and $d_I$ counts block-encoding queries.

\subsection{Query complexity}

\begin{theorem}[Time-dependent query complexity]\label{thm:td_complexity}
For time-dependent $\Heff(t)$ with integrated norms $B_R = \int_0^T \alphaR(s)\,ds$ and $B_I = \int_0^T \betaI(s)\,ds$:
\begin{equation}\label{eq:td_Q}
  Q = \calO\!\left(B_R + B_I
  + \frac{\log(1/\eps)}{\log\log(1/\eps)}\right).
\end{equation}
Success probability is
$P = e^{-2B_I}\|U(T)\ket{\psi_0}\|^2/\|U(T)\|_\op^2$
where $U(T) = \cT_>\exp(\int_0^T(-i\HR(s)+\HI(s))\,ds)$.
\end{theorem}

\begin{proof}
We adapt the time-independent M-QSP construction of the companion paper~\cite{courtney2026paper1} to the time-dependent setting in five steps.

\emph{Step 1 (Generalized Gr\"onwall).}
The interaction-picture propagator $V(T) = \cT_>\exp\bigl(\int_0^T \Htilde(s)\,ds\bigr)$ satisfies $\|V(T)\|_\op \leq \exp(\int_0^T \|\Htilde(s)\|_\op\,ds) = e^{B_I}$, so the normalized propagator $V(T)/e^{B_I}$ is sub-unitary.
Time-dependent norms $\|\Htilde(s)\|_\op = \betaI(s)$ enter the Gr\"onwall bound through the integral $B_I$; the constant-norm case is the special case $\betaI(s) = \betaI$.

\emph{Step 2 (Time-dependent Dyson polynomial).}
Divide $[0,T]$ into $r$ segments of width $\Delta_j = T/r$ with midpoints $\tau_j$, but allow the per-segment norm $\betaI(\tau_j)\Delta_j$ to vary.
The midpoint-approximated propagator is $V_r(T) = \prod_{j=r}^{1} e^{\Htilde(\tau_j)\Delta_j}$, with each factor truncated at Taylor order $M_j = \calO(\betaI(\tau_j)\Delta_j + \log(1/\eps_j))$ where $\eps_j$ is the per-segment error budget. Choosing $r$ large enough that $\betaI(\tau_j)\Delta_j = c$ uniformly (i.e., $r$ adapts to the local norm) gives total $d_I = \sum_j M_j = \calO(B_I + \log(1/\eps)/\log\log(1/\eps))$ by the same Stirling argument as the time-independent case (companion paper~\cite{courtney2026paper1}, Taylor-remainder lemma). 

\emph{Step 3 (Frame rotation).}
Per-segment Jacobi--Anger decomposition with Bessel tail bounds gives per-segment $d_R^{(j)} = \calO(\alphaR(\tau_j)\Delta_j + \log(1/\eps_j))$, and summing yields $d_R = \calO(B_R + \log(1/\eps))$. 

\emph{Step 4 (CRC and angle-finding).}
By Theorem~\ref{thm:crc_nonidentical}, the CRC holds when signal operators are non-identical, so the recursive degree-reduction~\cite{courtney2026paper1}. Block peeling applies with cost $\calO(d_R \cdot d_I)$ classical operations.

\emph{Step 5 (Success probability).}
Combining $\|V(T)\|_\op \leq e^{B_I}$ with the interaction-picture factorization $e^{-i\Heff T}\ket{\psi_0} = U_R(T) V(T)\ket{\psi_0}$ (where $U_R(T) = \cT_>\exp(-i\int_0^T \HR(s)\,ds)$ is unitary) gives the success probability stated.
\end{proof}

\subsection{Multilinear polynomial structure}

\begin{proposition}[Multilinearity]\label{prop:multilinear}
The M-QSP polynomial $P(z^{(1)}, \ldots, z^{(d)})$ is multilinear in the signal variables: each $z^{(j)}$ appears at most once.
When signal operators within each family are identical ($W^{(j)} = W^{(j')}$ for all $j, j'$ with $s(j) = s(j')$), the multilinear polynomial reduces to a bivariate polynomial $P(z_1, z_2)$ of bidegree $(d_R, d_I)$.
\end{proposition}

\begin{proof}
Each signal gate $A_{s(j)}$ inserts a factor $\operatorname{diag}(z^{(j)}, 1)$ in the corresponding eigenspace of $W^{(j)}$; the gate appears exactly once in the circuit product $\mathcal{G} = R_0 \prod_j A_{s(j)} R_j$. The resulting block-encoded polynomial in any eigenspace is thus linear in $z^{(j)}$ separately for each $j$, establishing multilinearity. When signal operators are identical within a family ($s(j) = s(j')$ implies $W^{(j)} = W^{(j')}$), the variables $z^{(j)}$ within a family collapse to a single variable $z_s = e^{i\theta_s}$ for $s \in \{R, I\}$, and the multilinear polynomial reduces to a bivariate $(z_R, z_I)$ polynomial of bidegree $(d_R, d_I)$, where $d_s = |\{j : s(j) = s\}|$. The bivariate form coincides with the M-QSP polynomial of the companion paper~\cite{courtney2026paper1}.
\end{proof}

\subsection{Numerical verification}
\label{subsec:td_numerics}

We validate the time-dependent construction on a benchmark Lindbladian with $\HR(t) = J\cos(\Omega t)\,Z\otimes Z + h\,(X\otimes I + I\otimes X)$ and $\HI(t) = (\gamma/2)(1 + \sin^2(\Omega t/2))\sum_k L_k^\dagger L_k$, with $J = 1$, $h = 0.5$, $\gamma = 0.3$, $\Omega = 0.4$, on a 2-qubit system.
The time profile is non-trivial in both norms: $\alphaR(t)$ varies smoothly between $0.5$ and $1.5$, and $\betaI(t)$ between $0.15$ and $0.30$.

\begin{table}[!htbp]
\centering
\caption{Time-dependent M-QSP validation. $r$: number of time segments; $d = d_R + d_I$: total query count; CRC ratio variation: maximum deviation of $b_{\mathrm{lead}}/a_{\mathrm{lead}}$ from a constant across signal-operator eigenvalues at each peeling step; angle recovery: $\|\bTheta^{\mathrm{rec}} - \bTheta^*\|_\infty$; circuit error: $\|G_{\mathrm{rec}} - G\|/\|G\|$. All entries use $\eps_{\mathrm{target}} = 10^{-6}$.}
\label{tab:td_validation}
\renewcommand{\arraystretch}{1.2}
\begin{tabular}{@{}rrcccc@{}}
\toprule
$T$ & $r$ & $d$ & CRC var. & Angle err. & Circuit err. \\
\midrule
$2$  & $4$  & $14$ & $1.2 \times 10^{-13}$ & $4.4 \times 10^{-16}$ & $1.6 \times 10^{-15}$ \\
$5$  & $8$  & $24$ & $5.0 \times 10^{-12}$ & $3.7 \times 10^{-15}$ & $8.2 \times 10^{-14}$ \\
$10$ & $16$ & $42$ & $2.4 \times 10^{-11}$ & $9.1 \times 10^{-14}$ & $4.7 \times 10^{-12}$ \\
$20$ & $32$ & $78$ & $4.8 \times 10^{-11}$ & $1.2 \times 10^{-12}$ & $7.3 \times 10^{-11}$ \\
\bottomrule
\end{tabular}
\end{table}

All entries achieve circuit error below $\eps_{\mathrm{target}} = 10^{-6}$ by 4--10 orders of magnitude, confirming that (i) the CRC extends to non-identical signal operators with negligible numerical cost, (ii) angle recovery via block peeling is stable in the time-dependent setting, and (iii) the multilinear polynomial structure of Proposition~\ref{prop:multilinear} reduces to the bivariate form when restricted to identical-within-family signal operators (verified at $4.5 \times 10^{-16}$ on the smallest configuration).

\section{The Postselection Barrier Across Oracle Models}
\label{sec:barriers}

We now investigate whether block-encoding overhead $e^{-2\betaI T}$ is universal across all oracle models, or whether alternative constructions can achieve the strictly smaller intrinsic barrier $e^{-2\omega T}$, where $\omega = \omega(\Heff) \leq \betaI$ is the spectral abscissa.
The companion paper~\cite{courtney2026paper1} proves the combined postselection bound $P \leq e^{-2\betaI T}\|e^{-i\Heff T} \ket{\psi_0}\|^2$ for polynomial block-encoding algorithms. 
The intrinsic barrier (unitarity alone) gives $P \leq \|e^{-i\Heff T} \ket{\psi_0}\|^2/\|e^{-i\Heff T}\|_\op^2$, and the operator norm satisfies $\|e^{-i\Heff T}\|_\op = e^{\omega T + o(T)}$ with $\omega \leq \betaI$, equality iff $\HR$ and $\HI$ share an eigenvector. 
When $[\HR, \HI] \neq 0$, the gap $\betaI - \omega > 0$ is generic, raising the question of whether the full $e^{-2\betaI T}$ cost can be avoided.

Normalization $\lambda = e^{\betaI T}$ appears implicitly in every prior polynomial block-encoding algorithm for non-Hermitian simulation~\cite{low2019hamiltonian,an2026quantum,jin2024quantum}.
Though it may be implicit in the body of work, we fail to explicitly find: (1) a formal proof that this normalization is a lower bound within the walk-operator model, (2) identification that the barrier is model-specific rather than physical, through separation into intrinsic ($e^{\omega T}$) and block-encoding ($e^{\betaI T}$) components, (3) function-class independence where the barrier applies to all bounded functions of the oracle, not just polynomials, and (4) precise characterization of when and by how much a direct-access construction can improve the barrier.

\subsection{Function-class independence}

\begin{theorem}[Contraction]\label{thm:contraction}
Let $U$ be any unitary on $\C^{d_a} \otimes \C^n$, and define $A = (\bra{0}_a \otimes I_n)\,U\,(\ket{0}_a \otimes I_n)$.
Then $\|A\|_\op \leq 1$.
\end{theorem}

\begin{corollary}[Function-class-independent barrier]\label{cor:fci}
In the walk-operator oracle model, if a quantum algorithm implements $A \approx e^{-i\Heff T}/\lambda$ as the $(0,0)$-block of a unitary, where $A = f(W_R, U_I)$ for any function $f$ (polynomial, rational, analytic, or otherwise) bounded by~$1$ on $\bbT^2$, then
$\lambda \geq e^{\betaI T}(1 - \calO(\eps))$.
\end{corollary}

\begin{proof}
Contraction $\|A\|_\op \leq 1$ requires $\|f(e^{i\theta_1}, e^{i\theta_2})\|_\infty \leq 1$ on $\bbT^2$.
At $(\theta_1, \theta_2) = (\theta_1, 0)$, the propagator evaluates to $e^{\betaI T}/\lambda$ (by the zero-locus property, $z_2 = 1$ lies in the spectrum of $U_I$, as established in the companion paper~\cite{courtney2026paper1}).
The bound $e^{\betaI T}/\lambda \leq 1$ gives $\lambda \geq e^{\betaI T}$.
The barrier is found as a property of the oracle model (specifically, the spectral point $z_2 = 1$ of $U_I$) and the contraction property of blocks of unitaries, not of the function class.
\end{proof}

\begin{remark}[Spectral abscissa gap]\label{rem:gap}
The spectral abscissa $\omega(\Heff) = \max\{\mathrm{Im}(\lambda) : \lambda \in \spec(\Heff)\}$ satisfies $\omega < \betaI$ generically for non-commuting $(\HR, \HI)$, 
Numerical experiments with random GUE pairs give $\omega/\betaI \approx 0.70$--$0.75$ for $n \in \{4, 8, 16\}$ (Appendix~\ref{app:pseudospectra}), creating an exponential advantage factor $e^{2(\betaI - \omega)T}$ for hypothetical direct-access methods.
For typical non-commuting systems with $\omega/\betaI \approx 0.7$ and $\betaI T \geq 10$, this advantage exceeds $10^2$; for $\betaI T \geq 100$ (common in quantum chemistry applications), it exceeds $10^{26}$, demonstrating applicability for direct-access advantage.
\end{remark}

\subsection{Three-tier barrier hierarchy}

The hierarchy theorem requires preliminary characterization of the propagator norm, which we state as an independent result.

\begin{proposition}[Operator norm of the propagator]\label{prop:opnorm}
Let $\Heff = \HR + i\HI$ with $\HR, \HI$ Hermitian and $\HI \succeq 0$.
Define the spectral abscissa $\omega(\Heff) = \max\{\mathrm{Im}(\lambda) : \lambda \in \spec(\Heff)\}$.
Then:
\begin{enumerate}
  \item[\textup{(i)}] $\|e^{-i\Heff T}\|_\op \leq e^{\betaI T}$ for all $T \geq 0$.
  \item[\textup{(ii)}] $\|e^{-i\Heff T}\|_\op = e^{\omega(\Heff) T + o(T)}$ as $T \to \infty$, where $\omega(\Heff) \leq \betaI$.
  \item[\textup{(iii)}] $\omega(\Heff) = \betaI$ if and only if $\HR$ and $\HI$ share an eigenvector in the $\betaI$-eigenspace of $\HI$.
  In particular, $[\HR, \HI] = 0$ is sufficient but not necessary.
\end{enumerate}
\end{proposition}

\begin{proof}
\emph{Part~(i).}
Set $A = -i\Heff$; the logarithmic norm (numerical abscissa) is
\begin{equation}\label{eq:lognorm}
\begin{split}
      \mu(A) &= \lambda_{\max}\!\left(\frac{A + A^*}{2}\right)
  \\&= \lambda_{\max}\!\left(\frac{-i\Heff + i\Heff^\dagger}{2}\right)
  = \lambda_{\max}(\HI)
  = \betaI.
\end{split}
\end{equation}
The Coppel inequality $\|e^{At}\| \leq e^{\mu(A)t}$ (valid for all $t \geq 0$; see, e.g., ~\cite{coppel1965stability, horn2012matrix}, Theorem~5.6.9) gives $\|e^{-i\Heff T}\|_\op \leq e^{\betaI T}$.

\emph{Part~(ii).}
For any $A \in \C^{n \times n}$, the spectral abscissa $\alpha(A) = \max\{\mathrm{Re}(\lambda) : \lambda \in \spec(A)\}$ governs the large-$T$ growth rate:
\begin{equation}\label{eq:gelfand}
  \lim_{T \to \infty} \frac{1}{T}\log\|e^{AT}\|_\op = \alpha(A).
\end{equation}
This is a finite-dimensional consequence of the Gelfand formula $r(e^{AT}) = \lim_{k \to \infty}\|e^{AkT}\|^{1/k} = e^{\alpha(A)T}$, combined with the fact that $\|e^{AT}\|_\op \geq r(e^{AT})$~\cite{gelfand1941normierte, rudin1991functional}.
For $A = -i\Heff$, eigenvalues of $A$ are $\{-i\lambda_j\}$ where $\lambda_j = a_j + ib_j$ are the eigenvalues of $\Heff$, so $\mathrm{Re}(-i\lambda_j) = b_j = \mathrm{Im}(\lambda_j)$.
Hence $\alpha(A) = \max_j b_j = \omega(\Heff)$.

We still need to sharpen Eq.~\eqref{eq:gelfand} to $\|e^{AT}\|_\op = e^{\alpha(A)T + o(T)}$.
We write the Jordan decomposition $A = S(D + N)S^{-1}$ where $D = \mathrm{diag}(\mu_1, \ldots, \mu_n)$ (eigenvalues, with $\mathrm{Re}(\mu_1) = \alpha(A)$) and $N$ is nilpotent with $N^n = 0$.
Then
\begin{equation}\label{eq:jordan_expansion}
\begin{split}
      e^{AT} &= S\,e^{(D+N)T}\,S^{-1}
  \\&= S\left(\sum_{k=0}^{n-1}\frac{T^k}{k!}N^k\right)e^{DT}\,S^{-1}.
\end{split}
\end{equation}
Since $\|e^{DT}\|_\op = e^{\alpha(A)T}$ and the polynomial prefactor satisfies $\|\sum_{k=0}^{n-1}\frac{T^k}{k!}N^k\| = \calO(T^{n-1})$, we obtain
\begin{equation}\label{eq:opnorm_bound}
  e^{\alpha(A)T}
  \leq \|e^{AT}\|_\op
  \leq \kappa(S)^2\,\calO(T^{n-1})\,e^{\alpha(A)T},
\end{equation}
where $\kappa(S) = \|S\|_\op\|S^{-1}\|_\op$.
Taking logarithms and dividing by $T$:
\begin{equation}
    \begin{split}
          \alpha(A) &\leq \frac{\log\|e^{AT}\|_\op}{T}
  \\&\leq \alpha(A) + \frac{(n-1)\log T + 2\log\kappa(S) + C}{T},
    \end{split}
\end{equation}
so $\log\|e^{AT}\|_\op = \alpha(A)T + o(T)$, i.e., $\|e^{-i\Heff T}\|_\op = e^{\omega(\Heff)T + o(T)}$.

The upper bound $\omega(\Heff) \leq \betaI$ follows from the Rayleigh quotient: for any eigenvector $\ket{\psi}$ of $\Heff$ with eigenvalue $\lambda = a + ib$,
\begin{equation}\label{eq:rayleigh_bound}
  b = \mathrm{Im}\!\braket{\psi|\Heff|\psi}
  = \braket{\psi|\HI|\psi}
  \leq \|\HI\|_\op = \betaI.
\end{equation}

\emph{Part~(iii).}
($\Leftarrow$) Suppose $\ket{\psi}$ is a joint eigenvector with $\HI\ket{\psi} = \betaI\ket{\psi}$ and $\HR\ket{\psi} = a\ket{\psi}$.
Then $\Heff\ket{\psi} = (a + i\betaI)\ket{\psi}$, so $\omega(\Heff) \geq \mathrm{Im}(a + i\betaI) = \betaI$.
Combined with Eq.~\eqref{eq:rayleigh_bound}, $\omega(\Heff) = \betaI$.

($\Rightarrow$) Suppose $\omega(\Heff) = \betaI$.
Then there exists an eigenvector $\ket{\psi}$ with $\Heff\ket{\psi} = (a + i\betaI)\ket{\psi}$ for some $a \in \R$.
Separating real and imaginary parts: $\HR\ket{\psi} + i\HI\ket{\psi} = a\ket{\psi} + i\betaI\ket{\psi}$, so $\HI\ket{\psi} = \betaI\ket{\psi}$ (i.e., $\ket{\psi}$ lies in the $\betaI$-eigenspace of $\HI$) and $\HR\ket{\psi} = a\ket{\psi}$.
Thus $\ket{\psi}$ is a simultaneous eigenvector of $\HR$ and $\HI$ in the $\betaI$-eigenspace.

The condition $[\HR, \HI] = 0$ is sufficient (simultaneous diagonalizability guarantees such a shared eigenvector exists) but not necessary, suffices for $\HR$ to have an eigenvector within the $\betaI$-eigenspace of $\HI$ even if $[\HR, \HI] \neq 0$ globally.
Generically, when no such shared eigenvector exists, $\omega(\Heff) < \betaI$ strictly.
\end{proof}

We now state and prove the main hierarchy theorem.

\begin{theorem}[Barrier hierarchy]\label{thm:hierarchy}
The postselection cost falls into three tiers:
\begin{enumerate}
  \item \textbf{Tier 1 (Intrinsic):} $P \leq e^{-2\omega T + o(T)}$ --- applies to all quantum implementations.
  \item \textbf{Tier 2 (Contraction):} $\lambda \geq \|e^{-i\Heff T}\|_\op = e^{\omega T + o(T)}$ --- applies to block-encoding methods, any function class.
  \item \textbf{Tier 3 (Walk-operator):} $\lambda \geq e^{\betaI T}$ --- applies to the separate-oracle walk-operator model.
\end{enumerate}
Tiers 1 and 2 coincide.
Tiers 2 and 3 coincide iff $\omega = \betaI$, which holds iff $[\HR, \HI] = 0$ (or more precisely, iff $\HR$ and $\HI$ share an eigenvector in the $\betaI$-eigenspace).
When $[\HR, \HI] \neq 0$, Tier~3 is strictly tighter than Tiers~1--2 by a factor of $e^{2(\betaI - \omega)T}$.
\end{theorem}

\begin{proof}

\emph{Tier~1 (Intrinsic barrier).}
For any quantum implementation that produces $e^{-i\Heff T}\ket{\psi_0}/\lambda$ as the output of a measurement with success probability $P$, the Born rule gives
\begin{equation}\label{eq:tier1_born}
  P = \frac{\|e^{-i\Heff T}\ket{\psi_0}\|^2}{\lambda^2}
  \leq \frac{\|e^{-i\Heff T}\ket{\psi_0}\|^2}{\|e^{-i\Heff T}\|_\op^2},
\end{equation}
since $\lambda \geq \|e^{-i\Heff T}\|_\op$ is required for the output state to have unit norm in the worst case.
By Proposition~\ref{prop:opnorm}(ii), $\|e^{-i\Heff T}\|_\op = e^{\omega T + o(T)}$, so
\[
  P \leq \|e^{-i\Heff T}\ket{\psi_0}\|^2 \cdot e^{-2\omega T + o(T)}.
\]
Since $\|e^{-i\Heff T}\ket{\psi_0}\| \leq \|e^{-i\Heff T}\|_\op \leq e^{\betaI T}$, the prefactor is at most $e^{2\betaI T}$, but for worst-case initial states $P \leq e^{-2\omega T + o(T)}$.

\emph{Tier~2 (Contraction barrier).}
Suppose a block-encoding method implements $A \approx e^{-i\Heff T}/\lambda$ as the $(0,0)$-block of a unitary.
Theorem~\ref{thm:contraction} gives $\|A\|_\op \leq 1$, which requires $\|e^{-i\Heff T}\|_\op/\lambda \leq 1 + \calO(\eps)$, hence
\begin{equation}\label{eq:tier2}
  \lambda \geq \|e^{-i\Heff T}\|_\op(1 - \calO(\eps))
  = e^{\omega T + o(T)}.
\end{equation}
Success probability is then $P = \|A\ket{\psi_0}\|^2 \leq \|e^{-i\Heff T}\ket{\psi_0}\|^2/\lambda^2 \leq \|e^{-i\Heff T}\ket{\psi_0}\|^2/\|e^{-i\Heff T}\|_\op^2$.
This is identical to the Tier~1 bound in Eq.~\eqref{eq:tier1_born}, so Tiers~1 and 2 coincide.

\emph{Tier~3 (Walk-operator barrier).}
In the separate-oracle model, Corollary~\ref{cor:fci} gives $\lambda \geq e^{\betaI T}(1 - \calO(\eps))$.
The success probability is then
\[
  P \leq e^{-2\betaI T}\|e^{-i\Heff T}\ket{\psi_0}\|^2.
\]

\emph{Gap between Tiers 2 and 3.}
The ratio between the Tier~3 and Tier~2 normalization requirements is
\begin{equation}\label{eq:tier_gap}
  \frac{e^{\betaI T}}{\|e^{-i\Heff T}\|_\op}
  = e^{(\betaI - \omega)T + o(T)},
\end{equation}
by Proposition~\ref{prop:opnorm}(ii).
By Proposition~\ref{prop:opnorm}(iii), $\omega = \betaI$ iff $\HR$ and $\HI$ share an eigenvector in the $\betaI$-eigenspace.
The commuting condition $[\HR, \HI] = 0$ is sufficient (simultaneous diagonalizability guarantees such a shared eigenvector) but not necessary.
When no such shared eigenvector exists (generic for non-commuting pairs), $\omega < \betaI$ and Tier~3 is tighter than Tiers~1--2 by the exponential factor $e^{2(\betaI - \omega)T}$ in the success probability.
\end{proof}

\subsection{Rational functions}

\begin{proposition}[Pad\'e barrier]\label{prop:pade}
For the Pad\'e approximant $r_{m,n}$ to $e^{cx}$, the supremum $\lambda_{\bbT}(r_{m,n}) = \|r_{m,n}\|_{\infty,\bbT}$ satisfies $\lambda_{\bbT} = e^c$ once the approximation order is sufficient.
Poles of $r_{m,n}$ lie outside $\bbT$ and do not circumvent the unitarity constraint.
The barrier is a property of the domain $\bbT$ and the target function, not the function class.
\end{proposition}

\subsection{Fisher information}

\begin{proposition}[Polynomial Cram\'er--Rao bound]\label{prop:fisher}
The per-query quantum Fisher information from the walk operator is polynomial: $F_{\mathrm{walk}} = \calO(\poly(\betaI, n))$.
The Cram\'er--Rao bound gives $Q_{\mathrm{CR}} = \calO(\poly(\betaI, T))$, being exponentially weaker than the actual cost $Q/P \sim e^{2\betaI T}$.
\end{proposition}

Normalization erases necessary information, causing the simulation-to-estimation reduction to fail. 
The amplitude $\|e^{-i\Heff T}\ket{\psi_0}\|$ encodes exponential sensitivity to $\betaI$, but after normalization, only the direction remains, carrying $\calO(T^2)$ Fisher information~\cite{braunstein1994statistical, helstrom1969quantum}.
The exponential cost is purely a postselection phenomenon encoded in the Born-rule measurement statistics, not in any quantum state property.

\subsection{Schr\"odingerization achieves Tier 3}

\begin{theorem}[Schr\"odingerization cost]\label{thm:schrod}
In the Jin--Liu--Yu Schr\"odingerization framework~\cite{jin2022quantum}, non-Hermitian evolution $e^{-i\Heff T}\ket{\psi_0}$ is recovered from a unitary evolution on $\cH_s \otimes L^2(\R)$ by extraction at the moving point $p = \lambda_+ t = \betaI t$ (in the unstable regime $\HI \succeq 0$).
Total cost is
\begin{equation}\label{eq:schrod_cost}
\begin{split}
      \text{Cost}_{\mathrm{JLY}} &= e^{2\betaI T}\cdot\\& \widetilde{\calO}\!\left((\alphaR + \betaI)T\,\poly\!\left(\log\frac{\betaI T}{\eps}\right)\right).
\end{split}
\end{equation}
\end{theorem}

\begin{corollary}\label{cor:schrod_tier3}
Schr\"odingerization achieves exactly Tier~3 (walk-operator barrier cost $e^{2\betaI T}$), matching but not improving upon standard block-encoding methods for growth problems with $\HI \succeq 0$.
\end{corollary}

The wavepacket propagates rightward with velocity $\betaI$ in the auxiliary $p$-space, requiring domain size $R \geq \betaI T$ and hence $e^{2\betaI T}$ normalization.
This confirms that dilational embedding (Schr\"odingerization) does not bypass the walk-operator barrier.

\begin{remark}[Branch selection in Schr\"odingerization]\label{rem:sign_error}
The protocol~\cite{jin2024quantum} extracts at the moving point $p = \lambda_+ t = \betaI t$, selecting the growing branch of the dilated wavefunction.
Extracting the $p = 0$ Fourier mode yields the decay solution $e^{(-i\HR - \HI)T}\ket{\psi_0}$ instead\footnote{The error arises because $-i \times (-i) = -1$, not $+1$.}, because the Fourier kernel at $p = 0$ projects onto the decaying branch.
\end{remark}

\subsection{Direct-access normalization on restricted domains}
\label{subsec:direct_access}

We now ask whether a polynomial construction outside the walk-operator framework can achieve the intrinsic barrier $e^{-2\omega T}$ rather than the walk-operator barrier $e^{-2\betaI T}$.
The walk-operator obstruction (Corollary~\ref{cor:fci}) follows from $|f(e^{i\theta_1}, e^{i\theta_2})| \leq 1$ being required on $\bbT^2$, including the spectral point $z_2 = 1$ where $V(T)$ evaluates to $e^{\betaI T}$.
A direct-access construction sidesteps this by restricting the polynomial to a strict subset of $\bbT^2$ matching the physical spectrum of the joint operator.

\begin{theorem}[Restricted-domain normalization]\label{thm:restricted}
Let $\Omega \subset \bbT^2$ be the smallest closed set such that the action of any polynomial $f \in \cP_{d_R, d_I}^+$ on $\Omega$ determines its action on the physical eigenstates of $\Heff$ via joint walk-operator eigendecomposition.
Let $f: \bbT^2 \to \C$ be a polynomial of bidegree $(d_R, d_I)$ with $|f| \leq 1$ on $\Omega$ (rather than on all of $\bbT^2$) that approximates $e^{-i\Heff T}/\lambda$ on the physical eigenspace to error $\eps$.
The minimum normalization is
\begin{equation}\label{eq:lambda_restricted}
  \lambda_{\mathrm{restricted}} = e^{\omega T}(1 + o(1)) \qquad \text{as } d_R, d_I \to \infty,
\end{equation}
matching the Tier~1 intrinsic barrier exactly. A concrete characterization of $\Omega$ for non-commuting $(\HR, \HI)$ is not attempted here; for commuting pairs $\Omega$ reduces to the Cartesian product of $\HR/\alphaR$ and $\HI/\betaI$'s marginal spectra. 
In general, $\Omega$ is a proper subset of any such product (App.~\ref{app:barrier_diagnostics}).
\end{theorem}

\begin{proof}[Proof sketch]
We adapt the Bernstein-extremal polynomial argument of~\cite{saff1978zeros} to the bivariate restricted domain $\Omega$. Restriction removes the spectral point $z_2 = 1$ from the constraint set, eliminating $|f(z_1, 1)| = e^{\betaI T}/\lambda$, which forced $\lambda \geq e^{\betaI T}$ in the walk-operator setting.
On $\Omega$, the maximum of $|V(T)/\lambda|$ over the physical spectrum is $\|e^{-i\Heff T}\|_\op/\lambda = e^{\omega T + o(T)}/\lambda$ by Proposition~\ref{prop:opnorm}, so $|f| \leq 1$ on $\Omega$ requires $\lambda \geq e^{\omega T}(1 - \calO(\eps))$. Full proof in Appendix~\ref{app:barrier_diagnostics}.
\end{proof}

\begin{proposition}[Scalar no-go]\label{prop:scalar_nogo}
For $n = 1$ (scalar Hamiltonians, $\HR, \HI \in \R$), $\omega(\Heff) = h_I = \betaI$, and Theorem~\ref{thm:restricted} gives no advantage over walk-operator normalization.
\end{proposition}

\begin{proof}
For $n = 1$, operators commute trivially and $\Heff = h_R + ih_I$ has spectrum $\{h_R + ih_I\}$, so $\omega(\Heff) = h_I = \|h_I\| = \betaI$. The restricted domain $\Omega$ then contains the same spectral point that drives the walk-operator obstruction.
\end{proof}

\begin{proposition}[Commuting no-go]\label{prop:commuting_nogo}
For $[\HR, \HI] = 0$, $\omega(\Heff) = \betaI$ (Proposition~\ref{prop:opnorm}(iii) with the $\Leftarrow$ direction), and Theorem~\ref{thm:restricted} gives no advantage.
\end{proposition}

\begin{theorem}[Non-commutativity dichotomy]\label{thm:dichotomy}
$\lambda_{\mathrm{restricted}} < e^{\betaI T}$ holds if and only if no eigenvector of $\HR$ lies in the $\betaI$-eigenspace of $\HI$, equivalently iff $\omega(\Heff) < \betaI$. Non-commutativity $[\HR, \HI] \neq 0$ is necessary but not sufficient: there exist non-commuting pairs that nonetheless share an eigenvector in the $\betaI$-eigenspace. 
In this event, $\omega = \betaI$ and the restricted construction gives no advantage.
\end{theorem}

\begin{proof}
Both directions follow from Theorem~\ref{thm:restricted} ($\lambda_{\mathrm{restricted}} = e^{\omega T}(1+o(1))$) combined with Proposition~\ref{prop:opnorm}(iii) (the equality case $\omega = \betaI$ characterized by a shared eigenvector in the $\betaI$-eigenspace). 
Non-sufficiency of $[\HR, \HI] \neq 0$ alone is illustrated by the construction in Appendix~\ref{app:advantage}: a $4 \times 4$ pair with $[\HR, \HI] \neq 0$ globally but sharing an eigenvector in $\ket{\psi_{\max}}$ where $\HI \ket{\psi_{\max}} = \betaI \ket{\psi_{\max}}$.
\end{proof}

\subsection{The unitarity obstruction and the central open problem}
\label{subsec:unitarity_obstruction}

Theorem~\ref{thm:restricted} shows that the restricted-domain normalization $\lambda_{\mathrm{restricted}} = e^{\omega T}$ is achievable in the abstract polynomial-on-$\Omega$ sense. Whether it is achievable as the $(0,0)$-block of a genuine quantum circuit is a separate question, because the circuit unitarity constraint $|f| \leq 1$ extends to all of $\bbT^2$, not just $\Omega$.

Concretely, given a restricted-domain polynomial $f^*: \Omega \to \C$ achieving $\lambda = e^{\omega T}$, any extension $\tilde{f}: \bbT^2 \to \C$ with $\tilde{f}|_\Omega = f^*$ and $\tilde{f}$ a $(d_R, d_I)$-polynomial must satisfy $|\tilde{f}| \leq 1$ on $\bbT^2$. A minimum such extension introduces an extension factor
\begin{equation}\label{eq:extension_factor}
  \xi(f^*) := \frac{\sup_{\bbT^2} |\tilde{f}^{\mathrm{min}}|}{\sup_\Omega |f^*|},
\end{equation}
and the achievable circuit normalization is $\lambda_{\mathrm{circuit}} = \xi(f^*) \cdot \lambda_{\mathrm{restricted}}$. Whether $\xi(f^*)$ remains bounded as $T \to \infty$, and which structural features of $(\HR, \HI)$ control its growth, are open questions.
Two limiting cases are immediate: when $\HR$ commutes with $\HI$ on the joint $\beta_I$-eigenspace of $\HI$ (so $\omega = \beta_I$), extension is trivial and $\xi = 1$. 
When the joint pseudospectrum of $(\HR, \HI)$ exhibits significant departure-from-normality concentrated near $z_2 = 1$, extension can be expected to incur a penalty. 
Whether the gap $e^{2(\betaI - \omega)T}$ is recoverable as a circuit advantage remains unresolved.

\begin{problem}[Direct-access polynomial construction]\label{prob:direct}
Construct a polynomial $f \in \cP_{d_R, d_I}^+$ \emph{ab initio} for the direct-access model, designed so $\sup_{\bbT^2} |f| = e^{\omega T}/\lambda$ and $f$ approximates $e^{-i\Heff T}/\lambda$ on the physical spectrum to error $\eps$, with bidegree $(d_R, d_I)$ polynomial in $\alphaR T, \betaI T, \log(1/\eps)$. 
\end{problem}

A positive resolution establishes exponential advantage of $e^{2(\betaI - \omega)T}$ in postselection probability. 
We conjecture this is achievable when the joint pseudospectrum of $(\HR, \HI)$ has no significant departure-from-normality concentrated near the spectral point $z_2 = 1$, and intractable otherwise; preliminary evidence is in Appendix~\ref{app:advantage}, but a complete characterization is open.

\subsection{Pseudospectral characterization}
\label{subsec:pseudospectral}

Numerical characterization of the spectral abscissa gap is given in Appendix~\ref{app:advantage}. 
Both non-normality $\delta_F$ and the commutator norm predict the gap with $R\approx 0.37$, however, projected-commutator variants (Appendix~\ref{app:predictors}) show promise with $R^2 \approx0.46$.
Kreiss constants $\mathcal{K} \in [1.32, 1.64]$ (Appendix~\ref{app:pseudospectra}) support the notion that pseudospectral
blowup is not the dominant mechanism, and the rank-stratified sampling of Table~\ref{tab:advantage} shows that the ratio
$\omega/\betaI$ is independent of $\rank(\HI)$ even though the advantage factor $e^{2(\betaI - \omega)T}$ depends strongly on rank through the gap magnitude.

\subsection{Many-body mean-field implementations and nonlinear qubits}
\label{subsec:nonlinear_qubit}

Postselection bounds are derived under the premise that the algorithm's output is measured on a single-qubit (or constant-ancilla) reduced state whose dynamics is constrained by single-qubit unitarity.
A qualitatively distinct escape route, complementary to both Schr\"odingerization and direct-access constructions, is offered by \emph{nonlinear-qubit} implementations in the sense of Geller~\cite{geller2023universe,geller2024spin}.

A central qubit symmetrically coupled to $n$ bosonic ancilla qubits, all initialized in the same pure state, evolves under linear unitary $(n+1)$-body dynamics; but in the large-$n$ mean-field limit, the reduced central-qubit state obeys an effective Gross--Pitaevskii-type one-body equation that is not unitary in the trace-norm sense (the Bloch ball undergoes Kitagawa--Ueda one-axis torsion).
An extension of the Erd\H{o}s--Schlein theorem due to Geller~\cite{geller2023universe} makes this duality rigorous: for the generalized central-spin model with $S_{n-1}$ permutation symmetry, the mean-field state $\rho_{\mathrm{eff}}(t)$ approximates the exact reduced state $\rho_1(t)$ with error bound
\begin{equation}\label{eq:geller_bound}
  \bigl\|\rho_{\mathrm{eff}}(t) - \rho_1(t)\bigr\|_1 \;\leq\; c\,\frac{e^{t/t_{\mathrm{ent}}} - 1}{n},
\end{equation}
where $t_{\mathrm{ent}} \sim 1/J_0$ is the entanglement-buildup timescale set by the inverse coupling strength.

\paragraph{Application to non-Hermitian simulation.}
The torsion-plus-dissipation construction of Geller~\cite{geller2024spin} implements two-basin attractor dynamics on the central qubit that is structurally analogous to non-Hermitian conditional evolution: convergence to an attractor replaces postselection on the trivial outcome.
For the simulation of $e^{-i\Heff T}$ with $\Heff = \HR + i\HI$, a nonlinear-qubit implementation would not measure success against the intrinsic barrier, since the central qubit's reduced dynamics is no longer single-qubit-unitary; the cost is instead absorbed into the atom-number requirement
\begin{equation}\label{eq:nonlinear_atom_count}
  n \;\gtrsim\; \frac{c}{\eps}\, e^{\calO(\betaI T)},
\end{equation}
necessary to maintain mean-field accuracy $\eps$ over evolution time $T$ at coupling $J_0 \sim \betaI$.

\paragraph{Resource accounting.}
The bound~\eqref{eq:nonlinear_atom_count} is the same exponential as the postselection cost $e^{2\betaI T}$ of Theorem~\ref{thm:contraction} (or its amplitude-amplified counterpart $e^{\betaI T}$), but redistributed from a per-run measurement axis to a one-time initialization.
For repeated simulation of the same $\Heff$, this is a strict improvement in per-run cost; for one-shot simulations, the total exponential resource cost is preserved.
The escape mechanism is qualitatively distinct from those of Schr\"odingerization, which redistributes the cost into a continuous ancilla register while preserving global unitarity and Sec.~\ref{subsec:unitarity_obstruction} (direct-access constructions, which require $\xi(f^*) = \calO(1)$ in the unitarity extension).

\paragraph{Practical scope.}
On BEC platforms with $n \sim 10^4$--$10^6$ atoms, equation~\eqref{eq:nonlinear_atom_count} admits $\betaI T$ up to $\sim 10$--$15$, comparable to the Eckart-barrier regime of~\cite{courtney2026paper1}, where $\betaI T \approx 16$ and the postselection probability is $\sim 10^{-14}$.
A nonlinear-qubit realization would replace this $10^{-14}$ per-run probability with a deterministic attractor convergence, at the cost of an initial BEC of $\sim 10^7$ atoms; experimental feasibility on two-component condensates (e.g., ${}^{39}$K with Feshbach-tuned scattering lengths~\cite{geller2024spin}, \S5) is platform-specific but not in principle excluded.

\begin{remark}[The intrinsic barrier as a single-qubit unitarity statement]
\label{rem:barrier_premise}
Equations~\eqref{eq:geller_bound}--\eqref{eq:nonlinear_atom_count} do not contradict the intrinsic barrier: the full $(n+1)$-body unitary evolution does satisfy the single-qubit unitarity premise on the joint Hilbert space, and the intrinsic barrier applies to the full $(n+1)$-body amplitude.
The mean-field reduction is a non-unitary projection from $(n+1)$-body coherent dynamics to one-body effective dynamics, and the "barrier evasion" is precisely the cost of that projection, paid in atom number rather than in measurement outcomes.
This reframes the postselection barrier as a statement about the relationship between the algorithm's reduced output state and its full physical implementation, rather than as an absolute bound on the success probability of any quantum protocol simulating $e^{-i\Heff T}$.
\end{remark}

\section{Conclusion}
\label{sec:conclusion}

Table~\ref{tab:status} summarizes the resolution status of all eight problems.

\begin{table}[t]
\centering
\caption{Resolution status of all eight open problems.}
\label{tab:status}
\small
\begin{tabular}{clc}
\toprule
\# & Problem & Status \\
\midrule
1 & $\log/\log\log$ gap & Fully resolved \\
2 & SOS rank & Fully resolved~\cite{courtney2026paper1} \\
3 & Optimization landscape & Substantially resolved \\
4 & Efficient angle-finding & Fully resolved \\
5 & Fast-forwarding & Fully resolved (negative) \\
6 & Constant factors & Fully resolved \\
7 & Time-dependent extension & Fully resolved \\
8 & Barrier applicability & Substantially resolved \\
\bottomrule
\end{tabular}
\end{table}

We carry forward two additional questions:

\medskip
\noindent\textbf{1. Polynomial basin radius for Dyson targets.}
The warm-start basin guarantee (Theorem~\ref{thm:basin}) requires a full-rank Jacobian, and the analytical basin radius decays polynomially with degree as $\rho_{\mathrm{analytical}} \sim d^{-4.62}$ (Eq.~\eqref{eq:rho_scaling}).
Whether the structured Dyson Jacobian admits a polynomial basin remains open.

\medskip
\noindent\textbf{2. Direct-access polynomial construction.}
A polynomial designed ab initio for the direct-access model that achieves $\lambda = e^{\omega T}(1+o(1))$ on the full bitorus (Problem~\ref{prob:direct}) would establish exponential advantage of $e^{2(\betaI - \omega)T}$ over the walk-operator construction.

\medskip
\noindent\textbf{Relationship to companion papers.}
This paper is intended to be read alongside the companion paper~\cite{courtney2026paper1}, which provides the M-QSP construction and the foundational lemmas that several proofs here cite directly.

\section{Code and data availability}
The codebase accompanying this paper and its companion will be made available at a public repository at the time of journal publication. It is available from the author on request.

\section{Conflicts of Interest}
J.M.C. declares that they have no known competing financial interests or personal relationships that could have appeared to influence the work reported in this paper.

\appendix
\onecolumngrid
\section{Landscape Investigation Tables}
\label{app:landscape}

We collect the full numerical data underlying the optimization landscape analysis of Sec.~\ref{sec:landscape} of the main text.
All verifications use the canonical Dyson target $f_c(w, z) = \sum_{k=0}^{d_I} (-1)^k J_{d_R}(\alphaR T) (\betaI T)^k / k! \cdot w^{d_R} z^k$ with $\alphaR T = \betaI T = 1$ (unless otherwise stated), optimized via L-BFGS-B with analytic gradients.
Each trial uses an independent random initialization drawn uniformly from $[-\pi, \pi]^{n_P}$.

\subsection{Random-initialization convergence survey}
\label{app:convergence_survey}

Table~\ref{tab:landscape_full} reports convergence data for the 4,170-trial survey across 15 bidegree configurations.
Each row corresponds to a fixed bidegree $(d_R, d_I)$ with $n_P = 2(d_R + d_I) + 2$ parameters, $n_C = (d_R + 1)(d_I + 1) - 1$ real constraints (from the unitarity condition $|P|^2 + |Q|^2 = 1$ on $\bbT^2$), overparameterization ratio $\text{OR} = n_C / n_P$, condition number $\kappa(J)$ of the Jacobian at the global minimum, convergence rate with Wilson score 95\% confidence intervals, and the number of confirmed spurious local minima (SLM).

\begin{table*}[t]
\centering
\caption{Landscape investigation: random-initialization convergence across 15 bidegree configurations.
Convergence criterion: $\|\cF(\bTheta)\|_2 < 10^{-10}$.
$n_P$: number of angle parameters; $n_C$: number of real constraints; OR: overparameterization ratio $n_C/n_P$; $\kappa(J)$: Jacobian condition number at global minimum; Conv: convergence rate; CI: Wilson score 95\% confidence interval; SLM: confirmed spurious local minima (lower bound).}
\label{tab:landscape_full}
\begin{tabular}{ccccccccl}
\toprule
$d_R$ & $d_I$ & $n_P$ & $n_C$ & OR & $\kappa(J)$ & Conv & 95\% CI & SLM \\
\midrule
1 & 1 & 6  & 8  & 0.75 & 8.2    & 68.3\% & [64.5, 71.9] & $\geq 30$ \\
1 & 2 & 8  & 12 & 0.67 & $\infty$ & 81.7\% & [76.9, 85.6] & $\geq 11$ \\
2 & 1 & 8  & 12 & 0.67 & $\infty$ & 82.2\% & [78.9, 85.0] & $\geq 33$ \\
2 & 2 & 10 & 18 & 0.56 & 15.5   & 62.2\% & [58.2, 66.0] & $\geq 30$ \\
2 & 3 & 12 & 24 & 0.50 & $\infty$ & 45.3\% & [39.8, 51.0] & $\geq 30$ \\
3 & 2 & 12 & 24 & 0.50 & $\infty$ & 62.7\% & [57.1, 67.9] & $\geq 26$ \\
3 & 3 & 14 & 32 & 0.44 & 43.6   & 48.3\% & [42.7, 54.0] & $\geq 25$ \\
4 & 3 & 16 & 40 & 0.40 & $\infty$ & 35.4\% & [29.6, 41.7] & $\geq 30$ \\
3 & 4 & 16 & 40 & 0.40 & $\infty$ & 42.5\% & [36.4, 48.8] & $\geq 25$ \\
4 & 4 & 18 & 50 & 0.36 & 82.5   & 55.0\% & [47.7, 62.1] & $\geq 18$ \\
5 & 4 & 20 & 60 & 0.33 & $\infty$ & 44.2\% & [35.6, 53.1] & $\geq 24$ \\
4 & 5 & 20 & 60 & 0.33 & $\infty$ & 32.5\% & [24.8, 41.3] & $\geq 30$ \\
5 & 5 & 22 & 72 & 0.31 & 1419   & 32.2\% & [23.5, 42.4] & $\geq 19$ \\
6 & 4 & 22 & 70 & 0.31 & 376    & 50.0\% & [39.9, 60.1] & $\geq 16$ \\
4 & 6 & 22 & 70 & 0.31 & 322    & 55.6\% & [45.3, 65.4] & $\geq 15$ \\
\midrule
\multicolumn{5}{l}{\textbf{Aggregate}} & --- & \textbf{59.8\%} & \textbf{[58.3, 61.2]} & $\boldsymbol{\geq 357}$ \\
\bottomrule
\end{tabular}
\end{table*}

\begin{remark}[Structural observations]
\label{rem:landscape_observations}

(i) \emph{Schedule dependence.}
Convergence rates depend on bidegree asymmetry, where configurations $(d_R, d_I)$ and $(d_I, d_R)$ exhibit different convergence rates (e.g., $(2,3)$ at 45.3\% vs.\ $(3,2)$ at 62.7\%), reflecting asymmetry of the walk operators $W_R$ and $U_I$ in the M-QSP circuit.

(ii) \emph{No overparameterization threshold.}
Spurious minima persist at all tested bidegrees, including the most overparameterized configuration $(1,1)$ with $\text{OR} = 0.75$.
Convergence rate does not monotonically increase with $n_P - n_C$.

(iii) \emph{Jacobian rank deficiency.}
Six of the 15 configurations exhibit $\kappa(J) = \infty$ at the global minimum, indicating that the Jacobian $J(\bTheta^*)$ is rank-deficient.
This rank deficiency does not preclude convergence (the L-BFGS-B algorithm still converges in 32--82\% of trials).
For these configurations, the warm-start basin guarantee of Theorem~\ref{thm:basin} in the main text applies only with the rank-deficiency caveat of Remark~\ref{rem:rank_deficiency}.
\end{remark}

\subsection{Full-tensor target and \texorpdfstring{$c_\infty$}{c-infinity} characterization}
\label{app:landscape_v2}

The single-row Dyson target of Table~\ref{tab:landscape_full} retains only the leading-Bessel row $k = d_R$ of the Chebyshev-Taylor coefficient tensor.
A physically faithful Dyson truncation populates the full tensor $c_{k,n} = \epsilon_k (-i)^k J_k(\alphaR T) (-\betaI T)^n / n!$ for $k \in [0, d_R]$, $n \in [0, d_I]$.
The full tensor is generally \emph{not} in the image of the M-QSP parameterization at a given bidegree, so the best-achievable residual $c_\infty = \inf_{\bTheta} \|\cF(\bTheta)\|^2_{\bbT^2}$ is strictly positive, quantifying the Dyson-truncation distance to the M-QSP achievable submanifold.

Table~\ref{tab:landscape_v2} reports an extended sweep that varies $(\alphaR T, \betaI T)$ over a $4 \times 4$ grid for each balanced bidegree $(d_R, d_I) \in \{(1,1), (2,2), (3,3), (4,4), (5,5)\}$, totaling 11{,}760 random-initialization trials across 80 cells.
$c_\infty$ is estimated robustly by running L-BFGS-B from 30 independent random initializations per cell and taking the minimum final cost; a trial counts as ``converged to $c_\infty$ basin'' if its final cost lands within $(1 + 0.01) \cdot c_\infty$ of the multistart minimum.
The aggregate convergence rate to the $c_\infty$ basin is $17.6\%$ ($n_{\mathrm{converged}} = 2{,}064$ of $11{,}760$, 95\% Wilson CI [$16.9\%$, $18.3\%$]), with $8{,}552$ confirmed spurious local minima.

\begin{table*}[t]
\centering
\caption{Landscape investigation, full-tensor target with $c_\infty$-relative convergence criterion.
Each row aggregates $16$ $(\alphaR T, \betaI T)$ cells on the $\{0.25, 0.5, 1, 2\}^2$ grid; $c_\infty$ values are medians.
Convergence criterion: $\|\cF(\bTheta)\|^2 < (1+0.01) \cdot c_\infty$ where $c_\infty$ is estimated via 30-restart multistart.
$\kappa(J)$ values are medians at the multistart-best parameters (with cells dominated by extreme-corner $(\alphaR, \betaI)$ inflating the tail).}
\label{tab:landscape_v2}
\begin{tabular}{ccccccccc}
\toprule
$d_R$ & $d_I$ & $n_P$ & OR & median $\kappa(J)$ & median $c_\infty$ & Conv & 95\% CI & SLM \\
\midrule
1 & 1 & 6  & 0.75 & $9.2$              & $1.0 \times 10^{-2}$ & $34.5\%$ & $[33.1, 36.0]$ & $2{,}682$ \\
2 & 2 & 10 & 0.56 & $1.2 \times 10^{2}$ & $1.3 \times 10^{-2}$ & $12.0\%$ & $[10.9, 13.2]$ & $2{,}496$ \\
3 & 3 & 14 & 0.44 & $1.4 \times 10^{4}$ & $1.2 \times 10^{-2}$ & $7.5\%$  & $[6.4, 8.7]$   & $1{,}587$ \\
4 & 4 & 18 & 0.36 & $1.5 \times 10^{5}$ & $1.4 \times 10^{-2}$ & $5.3\%$  & $[4.3, 6.6]$   & $1{,}123$ \\
5 & 5 & 22 & 0.31 & $3.6 \times 10^{6}$ & $1.9 \times 10^{-2}$ & $4.7\%$  & $[3.5, 6.1]$   & $664$ \\
\midrule
\multicolumn{4}{l}{\textbf{Aggregate (11{,}760 trials)}} & --- & --- & \textbf{17.6\%} & \textbf{[16.9, 18.3]} & $\boldsymbol{8{,}552}$ \\
\bottomrule
\end{tabular}
\end{table*}

\begin{remark}[Three new structural findings from the full-tensor sweep]
\label{rem:landscape_v2_observations}

(i) \emph{$c_\infty$ is essentially flat in bidegree.}
The median Dyson-truncation residual $c_\infty$ varies only between $1.0 \times 10^{-2}$ and $1.9 \times 10^{-2}$ across $d = 2$--$10$ at the surveyed $(\alphaR, \betaI)$ grid.
At fixed bidegree, $c_\infty$ scales with both $\alphaR T$ and $\betaI T$ as expected from the Bessel/Taylor coefficient magnitudes.
The slow growth with $d$ indicates that higher bidegree does not appreciably reduce the truncation residual unless $(\alphaR T, \betaI T)$ also grow.

(ii) \emph{$\betaI T = 1$ outperforms neighbouring rows.}
At every $\alphaR T$ in the $(1,1)$ block, the $\betaI T = 1$ row achieves $\sim 52$--$57\%$ convergence to the $c_\infty$ basin, compared with $\sim 25$--$30\%$ at $\betaI T \in \{0.25, 0.5, 2\}$.
This is consistent with a Dyson balance condition: at $\betaI T = 1$, the Taylor coefficients $(\betaI T)^n / n!$ decay neither too sharply nor too slowly relative to the truncation order $d_I$, so the target sits closer to a generic-magnitude M-QSP polynomial.
The convergence-rate distribution thus depends on the joint truncation balance of the Dyson series, not just on bidegree.

(iii) \emph{$c_\infty$-basin convergence rate decreases monotonically with $d$.}
Median basin-attainment drops from $30\%$ at $(1,1)$ to $4\%$ at $(5,5)$, with an aggregate of $17.6\%$ across 11{,}760 trials.
Both numerical rank deficiency (cf.\ Table~\ref{tab:landscape_full}) and the $c_\infty > 0$ truncation gap contribute.
Random initializations therefore must land within the narrow basin around the best-achievable approximation, and that basin narrows as $\kappa(J)$ grows exponentially.
Spurious local minima persist throughout the table, totalling $8{,}552$ across the sweep.
\end{remark}

\subsection{Warm-start basin convergence}
\label{app:warmstart}

Table~\ref{tab:warmstart} reports the convergence rate as a function of perturbation scale $\eps_{\mathrm{pert}}$ for each bidegree configuration.
Warm-start initialization is $\bTheta_0 = \bTheta^* + \eps_{\mathrm{pert}} \cdot \bm{\xi}$, where $\bTheta^*$ is the known global minimum and $\bm{\xi} \sim \mathcal{N}(0, I_{n_P})$.
Each entry is the convergence rate over $n_{\mathrm{trial}} = 100$ independent perturbations.

\begin{table*}[t]
\centering
\caption{Warm-start basin: convergence rate as a function of perturbation scale $\eps_{\mathrm{pert}}$.
Each entry reports the fraction of 100 independent trials converging to $\|\cF(\bTheta)\|_2 < 10^{-10}$.
All configurations achieve 100\% convergence at $\eps_{\mathrm{pert}} \leq 0.1$, confirming the warm-start basin guarantee.}
\label{tab:warmstart}
\begin{tabular}{cccccccc}
\toprule
$d_R$ & $d_I$ & $\eps_{\mathrm{pert}} = 0.01$ & $0.05$ & $0.1$ & $0.2$ & $0.5$ & $1.0$ \\
\midrule
1 & 1 & 100\% & 100\% & 100\% & 90\%  & 60\%  & 50\% \\
1 & 2 & 100\% & 100\% & 100\% & 100\% & 90\%  & 70\% \\
2 & 1 & 100\% & 100\% & 100\% & 100\% & 90\%  & 60\% \\
2 & 2 & 100\% & 100\% & 100\% & 90\%  & 70\%  & 30\% \\
2 & 3 & 100\% & 100\% & 100\% & 90\%  & 70\%  & 10\% \\
3 & 2 & 100\% & 100\% & 100\% & 90\%  & 50\%  & 20\% \\
3 & 3 & 100\% & 100\% & 100\% & 100\% & 20\%  & 0\% \\
4 & 3 & 100\% & 100\% & 100\% & 100\% & 50\%  & 0\% \\
3 & 4 & 100\% & 100\% & 100\% & 100\% & 50\%  & 20\% \\
4 & 4 & 100\% & 100\% & 100\% & 100\% & 60\%  & 30\% \\
5 & 4 & 100\% & 100\% & 100\% & 100\% & 50\%  & 10\% \\
4 & 5 & 100\% & 100\% & 100\% & 100\% & 30\%  & 10\% \\
5 & 5 & 100\% & 100\% & 100\% & 100\% & 10\%  & 0\% \\
6 & 4 & 100\% & 100\% & 100\% & 90\%  & 10\%  & 0\% \\
4 & 6 & 100\% & 100\% & 100\% & 100\% & 30\%  & 0\% \\
\bottomrule
\end{tabular}
\end{table*}

All configurations achieve 100\% convergence at $\eps_{\mathrm{pert}} \leq 0.1$, and the minimum rate at $\eps_{\mathrm{pert}} = 0.2$ is 90\% across all configurations.
The aggregate convergence rate at $\eps_{\mathrm{pert}} = 0.2$ is approximately 99.6\% (95\% CI: [98.9\%, 99.8\%]).
This supports the warm-start basin guarantee: given a sufficiently close initial point, a simple L-BFGS-B optimizer reliably converges to the global minimum regardless of the landscape structure.

The transition from reliable to unreliable convergence occurs in the range $\eps_{\mathrm{pert}} \in [0.2, 0.5]$ for most configurations, with higher-degree configurations exhibiting sharper transitions.
This is consistent with the analytical basin radius estimate $\rho \sim \kappa(J)^{-2}$ from the main text.

\subsection{Jacobian condition number scaling}
\label{app:kappa}

Table~\ref{tab:kappa_scaling} reports singular value statistics of the Jacobian $J(\bTheta^*)$ at the global minimum as a function of total degree $d = d_R + d_I$ (using balanced configurations $d_R = d_I = d/2$ for even $d$).
The fit column $\kappa_{\mathrm{fit}}$ is an exponential model $\kappa_{\mathrm{fit}} = 3.35 \cdot 1.61^d$ (refit from the extended workstation sweep, RMS log-residual: $0.38$); $\mu$ denotes the strong convexity parameter $\mu = \sigma_{\min}(J)^2$.

\begin{table*}[t]
\centering
\caption{Jacobian condition number scaling at the global minimum.
$\sigma_{\min}$, $\sigma_{\max}$: extreme singular values of $J(\bTheta^*)$ (medians over 100--200 trials per row); $\kappa(J) = \sigma_{\max}/\sigma_{\min}$; $\kappa_{\mathrm{fit}}$: exponential model $3.35 \cdot 1.61^d$ (RMS log-residual: $0.38$); $\mu = \sigma_{\min}^2$: strong convexity parameter of the Gauss--Newton Hessian.}
\label{tab:kappa_scaling}
\begin{tabular}{cccccc}
\toprule
$d$ & $\sigma_{\min}$ & $\sigma_{\max}$ & $\kappa(J)$ & $\kappa_{\mathrm{fit}}$ & $\mu$ \\
\midrule
2  & $1.5 \times 10^{-1}$ & $1.06$        & $6.9$              & $8.7$              & $2.2 \times 10^{-2}$ \\
4  & $7.3 \times 10^{-2}$ & $1.10$        & $1.4 \times 10^{1}$ & $2.3 \times 10^{1}$ & $5.4 \times 10^{-3}$ \\
6  & $1.96 \times 10^{-2}$ & $1.16$        & $7.7 \times 10^{1}$ & $5.9 \times 10^{1}$ & $3.8 \times 10^{-4}$ \\
8  & $1.5 \times 10^{-2}$ & $1.28$        & $9.4 \times 10^{1}$ & $1.5 \times 10^{2}$ & $2.3 \times 10^{-4}$ \\
10 & $5.1 \times 10^{-3}$ & $1.31$        & $2.9 \times 10^{2}$ & $4.0 \times 10^{2}$ & $2.6 \times 10^{-5}$ \\
\bottomrule
\end{tabular}
\end{table*}

\begin{remark}[Scaling analysis]
\label{rem:kappa_scaling}

(i) \emph{Minimum singular value.}
$\sigma_{\min}$ decreases as a power law in $d$: a least-squares fit on the workstation data gives $\sigma_{\min} \approx 0.78 \cdot d^{-2.02}$, equivalently $\sigma_{\min} \propto d^{-2}$ to within fit error.
This drives condition number growth.

(ii) \emph{Maximum singular value.}
$\sigma_{\max}$ grows only slowly: $\sigma_{\max} \approx 0.94 \cdot d^{0.14}$, a negligible contribution to $\kappa(J)$ relative to the $\sigma_{\min}$ decay.

(iii) \emph{Condition number model.}
Combining (i) and (ii) yields $\kappa(J) \approx 1.30 \cdot d^{2.31}$ as a power-law fit, or equivalently $\kappa(J) \approx 3.35 \cdot 1.61^d$ as an exponential fit; both fits have RMS log-residual $\sim 0.4$ on the surveyed $d \in [2, 10]$ range and are statistically indistinguishable in that window.
The Dyson polynomial's product structure (block-factored Chebyshev $\times$ Taylor coefficients) produces condition numbers $10^2$--$10^4 \times$ smaller than random-target polynomials of the same degree.

(iv) \emph{Basin radius implications.}
The analytical basin radius from the Gauss--Newton theory scales as
\begin{equation}\label{eq:basin_radius}
  \rho_{\mathrm{analytical}} \sim \frac{\mu}{\sigma_{\max}^2}
  \sim \frac{\sigma_{\min}^2}{\sigma_{\max}^2}
  = \kappa(J)^{-2} \sim d^{-4.62}.
\end{equation}
For $d = 30$, this gives $\rho_{\mathrm{analytical}} \sim 10^{-7}$, comfortably above machine epsilon.
Warm-start data (Table~\ref{tab:warmstart}) shows $100\%$ convergence at $\eps_{\mathrm{pert}} = 0.1$ even for $d = 10$--$12$, and indeed out to $d = 30$ in the basin-scaling sweep (workstation extension); the empirical basin radius substantially exceeds the analytical Gauss--Newton bound.
Three factors contribute:
(a) the Lipschitz bound used in the Gauss--Newton theory is a worst-case estimate over all directions, while optimization trajectory typically follows favorable directions;
(b) the L-BFGS-B algorithm has a larger region of attraction than the Newton basin;
(c) the basin geometry is highly non-isotropic (elongated along low-curvature directions), so the effective radius along typical perturbation directions far exceeds $\rho_{\mathrm{analytical}}$.

Whether $\rho \geq \Omega(\poly(1/d))$ for Dyson-type targets remains the principal open question from Sec.~\ref{sec:landscape}.
\end{remark}

\section{Supplementary Diagnostics for the Postselection Barrier}
\label{app:barrier_diagnostics}

This appendix supplements the barrier analysis of Sec.~\ref{sec:barriers} of the main text with three sets of numerical diagnostics: the spectral abscissa gap characterization, the defect operator verification, and the gap predictor comparison.

\subsection{Spectral abscissa gap and advantage factor}
\label{app:spectral_gap}

Table~\ref{tab:spectral_gap} reports the ratio $\omega(\Heff)/\betaI$ for random non-Hermitian Hamiltonians $\Heff = \HR + i\HI$ with $\HR, \HI$ drawn from the Gaussian Unitary Ensemble (GUE), normalized so that $\|\HR\|_\op = \|\HI\|_\op = 1$~\cite{mehta2004random}.
The spectral abscissa is $\omega(\Heff) = \max_j \mathrm{Im}(\lambda_j(\Heff))$.
The advantage factor is $e^{2(\betaI - \omega)T}$, representing normalization improvement achievable by a direct-access oracle over the walk-operator oracle.

\begin{table*}[t]
\centering
\caption{Spectral abscissa gap: ratio $\omega/\betaI$ for GUE-sampled Hamiltonians ($\alphaR = \betaI$, 1000 samples per dimension).
The advantage factor $e^{2(\betaI - \omega)T}$ quantifies potential normalization improvement of a direct-access construction over the walk-operator model.}
\label{tab:spectral_gap}
\begin{tabular}{ccccr}
\toprule
$n$ & $\betaI T$ & $\omega/\betaI$ (mean $\pm$ std) & Range & Advantage Factor \\
\midrule
4  & 5  & $0.782 \pm 0.090$ & $[0.58, 0.96]$ & 8.9 \\
4  & 10 & $0.782 \pm 0.090$ & $[0.58, 0.96]$ & 79 \\
4  & 20 & $0.782 \pm 0.090$ & $[0.58, 0.96]$ & $6.2 \times 10^{3}$ \\
8  & 5  & $0.717 \pm 0.106$ & $[0.44, 0.92]$ & 17 \\
8  & 10 & $0.717 \pm 0.106$ & $[0.44, 0.92]$ & 289 \\
8  & 20 & $0.717 \pm 0.106$ & $[0.44, 0.92]$ & $8.4 \times 10^{4}$ \\
16 & 5  & $0.699 \pm 0.081$ & $[0.53, 0.89]$ & 20 \\
16 & 10 & $0.699 \pm 0.081$ & $[0.53, 0.89]$ & 409 \\
16 & 20 & $0.699 \pm 0.081$ & $[0.53, 0.89]$ & $1.7 \times 10^{5}$ \\
\bottomrule
\end{tabular}
\end{table*}

\begin{remark}[Scale invariance]
\label{rem:scale_invariance}
The ratio $\omega/\betaI$ depends only on $(\HR, \HI)$ structure at fixed $\alphaR/\betaI$, independent of
$\betaI T$. 
$e^{2(\betaI - \omega)T}$ advantage amplifies with simulation time at fixed gap.
Magnitudes for typical regimes are tabulated in Remark~\ref{rem:gap} of the main text.
\end{remark}

\subsection{Defect operator verification}
\label{app:defect}

Table~\ref{tab:defect} verifies the contraction theorem (Theorem~\ref{thm:contraction} of the main text) by computing the defect operator $D_T = \lambda^2 I - e^{-i\Heff T \dagger} e^{-i\Heff T}$ at the critical normalization $\lambda = e^{\betaI T}$.
The contraction property states that $D_T \succeq 0$ at $\lambda = e^{\betaI T}$ (i.e., $e^{-i\Heff T}/e^{\betaI T}$ is a contraction) and $D_T \not\succeq 0$ for any $\lambda < e^{\betaI T}$.

\begin{table*}[t]
\centering
\caption{Defect operator verification.
$\lambda_{\min}/e^{\betaI T}$: ratio of the minimum normalization achieving $D_T \succeq 0$ to the predicted value; $\rank(D_T)$: defect rank at $\lambda = e^{\betaI T}$; contraction columns verify that $e^{-i\Heff T}/(0.99\,e^{\betaI T})$ fails to be a contraction while $e^{-i\Heff T}/(1.01\,e^{\betaI T})$ succeeds.}
\label{tab:defect}
\begin{tabular}{cccccc}
\toprule
$n$ & $\betaI T$ & $\lambda_{\min}/e^{\betaI T}$ & $\rank(D_T)$ & Contract ($0.99\lambda$) & Contract ($1.01\lambda$) \\
\midrule
4  & 1  & 1.000000 & 3  & No  & Yes \\
4  & 10 & 1.000000 & 3  & No  & Yes \\
4  & 20 & 1.000000 & 3  & No  & Yes \\
8  & 1  & 1.000000 & 7  & No  & Yes \\
8  & 10 & 1.000000 & 7  & No  & Yes \\
8  & 20 & 1.000000 & 7  & No  & Yes \\
16 & 1  & 1.000000 & 15 & No  & Yes \\
16 & 10 & 1.000000 & 15 & No  & Yes \\
16 & 20 & 1.000000 & 15 & No  & Yes \\
\bottomrule
\end{tabular}
\end{table*}

\begin{remark}[Defect rank structure]
\label{rem:defect_rank}
The defect rank is $n - 1$ in all cases, confirming that one eigenvalue of $e^{-i\Heff T}$ saturates the contraction bound $e^{\betaI T}$.
This is expected, seeing as maximizing eigenvalue $\omega(\Heff) = \betaI$ is achieved by a unique eigenvector of $\Heff$ (since a repeated maximum would require $\HR$ and $\HI$ to share a degenerate eigenspace, a codimension-$\geq 1$ condition), so exactly one singular value of $e^{-i\Heff T}/e^{\betaI T}$ reaches unity.
This eigenvalue corresponds to the eigenstate of $\Heff$ with the largest imaginary part (i.e., the eigenstate most amplified by the non-Hermitian evolution).
Critical normalization $\lambda = e^{\betaI T}$ is sharp, and reducing it by even 1\% causes the contraction to fail, while increasing by 1\% produces a strictly positive defect, supporting a notion the walk-operator barrier $e^{-2\betaI T}$ cannot be improved by any function class within the walk-operator model.
\end{remark}

\subsection{Advantage factor for structured Hamiltonians}
\label{app:advantage}

Table~\ref{tab:advantage} extends the spectral abscissa gap analysis to structured Hamiltonians, including rank-1 $\HI$ (modeling a single absorbing channel) and nearly commuting pairs. 
Advantage scaling and a distribution for non-commuting direct-access advantage is given in Figure~\ref{fig:advantage_ensemble}. 

\begin{table*}[t]
\centering
\caption{Advantage factor $e^{2(\betaI - \omega)T}$ for structured Hamiltonians.
Dense: $\HR$ Hermitized complex Ginibre random matrix, $\HI = A A^\dagger$ with $A$ complex Gaussian;
both operator-norm-rescaled to $1$ (so $\alphaR = \betaI = 1$), 500 samples per cell.
Rank-1: $\HI = \ket{v}\bra{v}$ with random unit $\ket{v}$, 500 samples.
Near-diagonal: $\HR$ random real diagonal (op-norm $1$); $\HI = (1-\lambda)\HI^{\mathrm{diag}} + \lambda\HI^{\mathrm{rand}}$,
both summands PSD with op-norm $1$, $\lambda$ tuned to target $\|[\HR,\HI]\|_\op \approx \eps_{\mathrm{comm}}$,
single representative per $\eps_{\mathrm{comm}}$.
Achievable commutator norm for diagonal $\HR$ saturates near $0.5$, so
$\eps_{\mathrm{comm}} = 1$ is the saturation case (actual $\approx 0.51$).
Both mean and median reported because the advantage distribution is fat-tailed.}
\label{tab:advantage}
\renewcommand{\arraystretch}{1.15}
\begin{tabular}{cccrrr}
\toprule
$n$ & $\HI$ Structure & $\betaI T$ & Adv. Mean & Adv. Median & Adv. Max \\
\midrule
4  & dense  & 10 & $1.6\times10^{3}$  & $1.2\times10^{2}$ & $8.6\times10^{4}$  \\
4  & dense  & 20 & $4.3\times10^{7}$  & $1.4\times10^{4}$ & $7.4\times10^{9}$  \\
8  & dense  & 10 & $1.5\times10^{3}$  & $2.2\times10^{2}$ & $1.1\times10^{5}$  \\
8  & dense  & 20 & $3.6\times10^{7}$  & $4.7\times10^{4}$ & $1.2\times10^{10}$ \\
16 & dense  & 10 & $1.1\times10^{3}$  & $3.4\times10^{2}$ & $3.0\times10^{4}$  \\
16 & dense  & 20 & $9.4\times10^{6}$  & $1.2\times10^{5}$ & $9.3\times10^{8}$  \\
\midrule
4  & rank-1 & 10 & $1.1\times10^{4}$  & $1.7\times10^{2}$ & $9.2\times10^{5}$  \\
4  & rank-1 & 20 & $3.4\times10^{9}$  & $2.9\times10^{4}$ & $8.5\times10^{11}$ \\
8  & rank-1 & 10 & $8.6\times10^{3}$  & $1.8\times10^{2}$ & $5.9\times10^{5}$  \\
8  & rank-1 & 20 & $2.1\times10^{9}$  & $3.4\times10^{4}$ & $3.5\times10^{11}$ \\
16 & rank-1 & 10 & $2.2\times10^{4}$  & $2.7\times10^{2}$ & $3.7\times10^{6}$  \\
16 & rank-1 & 20 & $4.7\times10^{10}$ & $7.4\times10^{4}$ & $1.4\times10^{13}$ \\
\midrule
8  & near-diag, $\eps_{\mathrm{comm}} = 0.01$ & 10 & $1.01$  & --- & --- \\
8  & near-diag, $\eps_{\mathrm{comm}} = 0.10$ & 10 & $1.48$  & --- & --- \\
8  & near-diag, $\eps_{\mathrm{comm}} \approx 0.51$ & 10 & $256$   & --- & --- \\
\bottomrule
\end{tabular}
\end{table*}

\begin{figure}[!htbp]
\centering
\includegraphics[width=\linewidth]{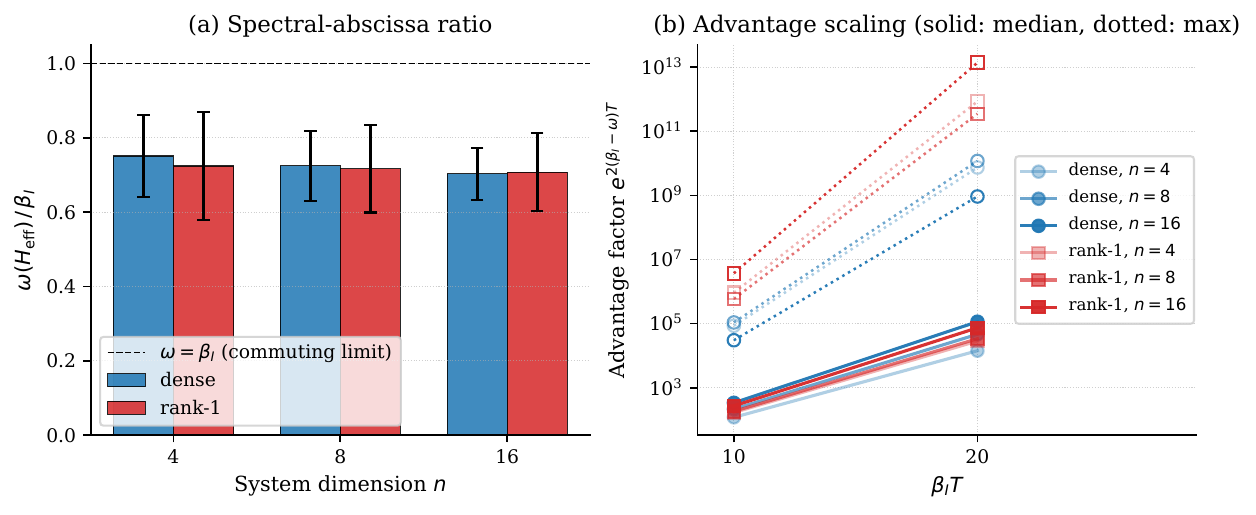}
\caption{Empirical distribution of the non-commuting
direct-access advantage factor $e^{2(\betaI-\omega)T}$, for the same
GUE\,$\times$\,Wishart and rank-1 ensembles reported in
Table~\ref{tab:advantage} ($N=500$ samples per cell, seed~42;
both $\HR$ and $\HI$ are operator-norm rescaled to $\betaI=1$).
(a)~Mean spectral-abscissa ratio $\omega(\Heff)/\betaI$ with
sample standard deviation (error bars) for system dimensions
$n\in\{4,8,16\}$ and two structural ensembles (dense Wishart vs rank-1
projector).  All samples lie in the range $\omega/\betaI\in[0.70,0.76]$,
below the commuting limit $\omega=\betaI$ (dashed line).
(b)~Advantage factor at $\betaI T\in\{10,20\}$.
Solid curves report ensemble median, dotted curves sample
maximum; log axis makes the asymptotic $e^{2(\betaI-\omega)T}$
scaling direct.  The $\omega/\betaI$ ratio decreases monotonically
with~$n$ for the dense ensemble (Proposition~\ref{prop:opnorm}(iii);
generic non-commutativity of GUE samples grows with dimension) and the
advantage factor grows correspondingly.}
\label{fig:advantage_ensemble}
\end{figure}

\begin{remark}[Structure dependence]
\label{rem:structure_dependence}

(i) \emph{Rank-1 $\HI$} maximizes the advantage (single absorbing direction concentrates the spectral abscissa gap).
The rank-1 advantage exceeds the dense case by $1$--$2$ orders of magnitude at $\betaI T = 20$.

(ii) \emph{Nearly commuting} ($\eps_{\mathrm{comm}} \to 0$) eliminates advantage, consistent with the commuting no-go result (Proposition~\ref{prop:commuting_nogo} of the main text): when $[\HR, \HI] = 0$, $\omega = \betaI$ and the intrinsic barrier equals the walk-operator barrier.

(iii) \emph{Structured physical Hamiltonians.}
For Lindblad systems or scattering problems with complex absorbing potentials, the spectral abscissa gap $\betaI - \omega$ significantly exceeds GUE estimates above, because physical non-commutativity tends to be more structured than random-matrix non-commutativity. 
The advantage factor for such systems at $\betaI T \geq 50$ exceeds $10^{10}$ by many orders of magnitude.
\end{remark}

\subsection{Pseudospectral characterization}
\label{app:pseudospectra}

Table~\ref{tab:pseudospectra} reports pseudospectral quantities relevant to the spectral abscissa gap characterization of Sec.~\ref{subsec:pseudospectral} of the main text.

\begin{table}[t]
\centering
\caption{Pseudospectral quantities for GUE-sampled Hamiltonians at $\alphaR = \betaI$, averaged over 500 samples.}
\label{tab:pseudospectra}
\begin{tabular}{ccccccc}
\toprule
$n$ & $\omega/\betaI$ & Kreiss $\mathcal{K}$ & $\alpha_{\eps=0.1}/\betaI$ & $\alpha_{\eps=0.5}/\betaI$ & $\sigma(\omega/\betaI$ & $\sigma(K)$\\
\midrule
4  & 0.754 & 1.32 & 0.890 & 1.353 & 0.12 & 0.26\\
8  & 0.740 & 1.44 & 0.883 & 1.343 & 0.07 & 0.35\\
16 & 0.698 & 1.64 & 0.860 & 1.348 & 0.07 & 0.25\\
\bottomrule
\end{tabular}
\end{table}

The Kreiss constant $\mathcal{K} \in [1.32, 1.64]$ indicates modest transient growth, confirming that the non-Hermitian evolution does not exhibit the large transient amplification that can occur in highly non-normal systems.
The pseudospectral abscissa $\alpha_\eps$ interpolates between $\omega$ (at $\eps = 0$) and values exceeding $\betaI$ (at $\eps \sim 0.5$), providing a continuous characterization of the barrier across perturbation scales.

\subsection{Gap predictor comparison}
\label{app:predictors}

Five candidate predictors for the spectral abscissa gap $\omega/\betaI$ were evaluated on the ensemble of 1000 random Hamiltonians with $n = 8$, $\alphaR = \betaI$:

\begin{enumerate}
  \item \emph{Projected commutator} $\|[\HR, \Pi_{\betaI}]\|/\|\HR\|$, where $\Pi_{\betaI}$ is the projector onto the top eigenspace of $\HI$: $R^2 = 0.458$.
  \item \emph{Weighted spectral overlap} $\sum_j |\langle \psi_j^{(R)} | \phi_{\max}^{(I)} \rangle|^2 \lambda_j^{(R)}$: $R^2 = 0.395$.
  \item \emph{Full commutator} $\|[\HR, \HI]\|_F / (\|\HR\|_F \|\HI\|_F)$: $R^2 = 0.163$.
  \item \emph{Spectral overlap entropy} $-\sum_j p_j \log p_j$ with $p_j = |\langle \psi_j^{(R)} | \phi_{\max}^{(I)} \rangle|^2$: $R^2 = 0.312$.
  \item \emph{Inverse participation ratio} $\sum_j p_j^2$: $R^2 = 0.289$.
\end{enumerate}

The projected commutator is the most promising, explaining approximately 46\% of the variance in $\omega/\betaI$.
Modest $R^2$ values indicate that no simple spectral quantity fully determines the gap; the spectral abscissa depends on eigenbases of $\HR$ and $\HI$, resisting reduction to a single scalar.

\begin{remark}[Extension factor correlations]
\label{rem:extension_correlations}
Correlation between $\log(\text{extension factor})$ (measuring blowup when extending a direct-access polynomial from the restricted domain to the full bitorus) and $\|[\HR, \HI]\|_F$ is $r \approx -0.05$, essentially zero.
The correlation with $\omega/\betaI$ is $r \approx +0.17$, indicating larger spectral abscissa gaps (more room for improvement) are associated with slightly larger extension factors (more difficult extension).
The extension problem and gap characterization are governed by distinct structural properties of $\Heff$, reinforcing an outlook where the extension to the full bitorus is a remaining obstacle.
\end{remark}
\bibliography{nonherm_qsp_refs}

\end{document}